\documentclass[12pt]{amsart}
\usepackage{a4wide,amsmath,amssymb,graphicx}
\usepackage{color}

\let\pa\partial
\let\na\nabla
\let\eps\varepsilon
\newcommand{\R}{\mathbb{R}}
\newcommand{\N}{\mathbb{N}}
\newcommand{\diver}{\textnormal{div}}

\newtheorem{theorem}{Theorem}
\newtheorem{proposition}[theorem]{Proposition}
\newtheorem{lemma}[theorem]{Lemma}
\theoremstyle{remark}

\begin{document}
\title[Two spinorial drift-diffusion models]{Two spinorial drift-diffusion models
for quantum electron transport in graphene}

\author{Nicola Zamponi}
\address{Dipartimento di Matematica Ulisse Dini, Viale Morgagni 67/A, 
Firenze, Italy}
\email{nicola.zamponi@math.unifi.it}

\author{Ansgar J\"ungel}
\address{Institute for Analysis and Scientific Computing, Vienna University of  
	Technology, Wiedner Hauptstra\ss e 8--10, 1040 Wien, Austria}
\email{juengel@tuwien.ac.at} 

\date{\today}

\thanks{The second author acknowledges partial support from   
the Austrian Science Fund (FWF), grants P20214, P22108, and I395, and
from the Austrian-French Project of the Austrian Exchange Service (\"OAD).} 

\begin{abstract}
Two drift-diffusion models for the quantum transport of electrons in gra\-phene,
which account for the spin degree of freedom, are derived from a spinorial
Wigner equation with relaxation-time or mass- and spin-conserving matrix collision
operators using a Chapman-Enskog expansion around the thermal equilibrium.
Explicit models are computed by assuming that both the semiclassical
parameter and the scaled Fermi energy are sufficiently small.
For one of the models, the global existence of weak solutions, entropy-dissipation
properties, and the exponential long-time decay of the spin vector are proved.
Finally, numerical simulations of a one-dimensional ballistic diode using both models are presented, showing the temporal behavior of the particle density and the 
components of the spin vector.

\end{abstract}

% \paragraph{Keywords:}  
 
% \paragraph{AMS classification:}  

\maketitle

%%%%%%%%%%%%%%%%%%%%%%%%%%%%%%%%%%%%%%%%%%%%%%%%%%%%%%%%%%%%%%%%%%%%%%%%%%%%%%

\section{Introduction}

Graphene is a new semiconductor material, which is a subject of great interest 
for nanoscale electronic applications. The reason for this interest is
due to the very remarkable properties of graphene,
such as the high electron mobility and long coherence length. Therefore,
graphene is a promising candidate for the construction of a new generation of 
electronic devices with far better performances than current silicon
devices \cite{Fre08}. 
Potential applications include, for instance, 
spin field-effect transistors \cite{JPTJW08,XFLA10},
extremely sensitive gas sensors \cite{SGMHBKN07},
one-electron graphene transistors \cite{PSKYHNG08},
and graphene spin transistors \cite{CCF07}.

Physically,
graphene is a two-dimensional semiconductor with a zero-width band gap,
consisting of a single layer of carbon atoms arranged in a honeycomb lattice.
In the energy spectrum, the valence band intersects the conduction band at some
isolated points, called the Dirac points. Around these points, 
quasiparticles in graphene exhibit the linear dispersion relation $E=v_F|p|$,
where $p$ denotes the crystal momentum and $v_F\approx 10^6$\,m/s is the Fermi
velocity \cite{MNG08}. This energy spectrum resembles the Dirac spectrum
for massless relativistic particles, $E=c|p|$, where $c$ is the speed of light.
Hence, the Fermi velocity $v_F\approx c/300$ takes the role of the speed of light.
The system Hamiltonian can be approximated near a Dirac point, 
for low energies and in the absence of a potential, by the Dirac-like operator
\begin{equation}\label{H0}
  H_0 = -i\hbar v_F\left(\sigma_1\frac{\pa}{\pa x_1}+\sigma\frac{\pa}{\pa x_2}\right),
\end{equation}
where $\hbar$ is the reduced Planck constant, and $\sigma_1$ and $\sigma_2$ are
the Pauli matrices (see \eqref{pauli}). 

In order to understand and predict the charge carrier transport in graphene, 
transport models, which incorporate the spin degree of freedom, have to be devised. 
Theoretical models for spin-polarized transport involve fluid-type
drift-diffusion equations, kinetic transport equations, and Monte-Carlo
simulation schemes; see the references in \cite{PSP05}. A hierarchy of
fluiddynamic spin models was derived from a spinor Boltzmann transport equation 
in \cite{BeEl09}. Suitable matrix collision operators
were suggested and analyzed in \cite{PoNe11}. Drift-diffusion models for spin
transport were considered in several works; see, e.g., \cite{BaMe10,ElH08,Sai04}.
A mathematical analysis of spin drift-diffusion systems for the band densities
is given in \cite{Gli08}.

Fluiddynamic equations provide a compromise between physical accuracy and
numerical efficiency. Another advantage is that they contain already
the physically interesting quantities, such as the particle density, momentum,
and spin densities, whereas other models usually involve variables which do not
have an immediate physical interpretation, like wavefunctions, density operators,
and Wigner distributions. In the latter case, further computations have to be
made to obtain the quantities of physical interest.

In this work, we address the quantum kinetic and diffusion level of spin-polarized
transport in graphene. More precisely, starting from a spinorial Wigner equation,
we aim to derive via a moment method and a Chapman-Enskog expansion
macroscopic drift-diffusion
models for the particle density and spin vector. Furthermore, we prove the
global existence of weak solutions to one of these models and we illustrate
the behavior of the solutions in a ballistic diode by numerical experiments..

We note that there are only very few articles concerned with kinetic or
macroscopic transport models for graphene. In the physics literature, 
the focus is on transport properties such as the carrier mobility \cite{Gui08},
charged impurity and phonon scattering \cite{DAHR11}, 
and Klein tunneling \cite{Per09}. Wigner models were investigated in \cite{MoSc11}.
Starting from a Wigner equation, hydrodynamic spin models were derived in
\cite{ZaBa11}, and the work \cite{Zam12} is concerned with the derivation
of drift-diffusion models for the band densities. In contrast, we will work in
the present paper with all components of the spin vector. Furthermore,  
we provide a mathematical analysis of one of the models
and numerical simulations of both models.

In the following, we describe our approach and the main results.
>From the unique features of graphene follow that Fermi-Dirac statistics would be
more suitable to describe quantum transport in the material than Maxwell-Boltzmann
statistics, since the energy spectrum of the Hamiltonian \eqref{H0} is not
bounded from below. We overcome this problem by modifying the Hamiltonian $H_0$.
In fact, we assume that the system Hamiltonian is approximated by
the following operator which is bounded from below:
\begin{equation}\label{H}
  H = H_0 - \left(\frac{\hbar^2}{2m}\Delta \right)\sigma_0,
\end{equation}
where $m>0$ is a parameter with the dimension of a mass and $\sigma_0$ is the
unit matrix in $\R^{2\times 2}$ (see \eqref{pauli}).
This is not a very restrictive assumption since the operator \eqref{H0} is itself
only an approximation of the correct system Hamiltonian, valid for small values
of the momentum $|p|$. 

Starting from the Hamiltonian $H+V\sigma_0$, where $V$ is the electric potential,
Wigner equations were derived from the Von-Neumann equation in \cite{ZaBa11}.
In order to derive diffusion models, we consider two types of collision
operators in the Wigner model. 

First, we employ relaxation-time terms of
BGK-type, $Q(w) = (g-w)/\tau_c$, where $w$ is the Wigner distribution,
$\tau_c>0$ is the mean free path,
and $g$ is the thermal equilibrium distribution derived from the quantum minimum
entropy principle \cite{DeRi03}; see Sections \ref{sec.kinetic} and \ref{sec.g}. 
We assume that the wave energy is much smaller
than the typical kinetic energy (semiclassical hypothesis) and that the
the scaled Planck constant is of the same order as the scaled Fermi energy
(low scaled Fermi speed hypothesis). Performing a diffusive limit and a
Chapman-Enskog expansion around the equilibrium distribution formally 
yields the first quantum spin diffusion model (QSDE1)
for the particle density $n_0$ and the spin vector $\vec n=(n_1,n_2,n_3)$,
which are the zeroth-order moments of the Wigner distribution $w$:
\begin{align*}
  & \pa_t n_0 - \diver\,J_0 = 0, \quad J_0 = \na n_0 + n_0 \na V, \\
  & \pa_t n_j - \diver\,J_j = F_j, \quad 
  J_j = A_0(\vec n/n_0)\na\vec n + \vec n\otimes\na V + B_0(n_0,\vec n), \quad
  j=1,2,3,
\end{align*}
where $A_0$ and $B_0$ are some functions,
and $F_j$ depends on $n_0$, $\vec n$, $\na\vec n$, and $\na V$.
We refer to Section \ref{sec.model1} for details.
Note that the equations for 
the particle density and spin vector decouple; only the spin vector
equation depends nonlinearly on $n_0$ and $\vec n$. 
The functions $A_0$ and $B_0$ are well defined only if $0\le|\vec n|/n_0<1$. Hence,
the main difficulty in the analysis of this model is the proof of lower and
upper bounds for $|\vec n|/n_0$.

Second, we employ a mass- and spin-conserving matrix collision operator suggested
in \cite{PoNe11} for a semiconductor subject to a magnetic field. Performing
a diffusive limit and a Chapman-Enskog expansion similarly as for the first model, 
we derive the second quantum spin diffusion model (QSDE2) in which the equations
for the particle density $n_0$ and the spin vector $\vec n$ are fully coupled:
\begin{align*}
  & \pa_t n_0 - \diver\,J_0 = 0, \quad 
  J_0 = A_1(\na n_0 + n_0\na V) + B_1\cdot(\na\vec n+\vec n\otimes\na V)
  + C_1(n_0,\vec n), \\
  & \pa_t n_j - \diver\,J_j = G_j, \quad
  J_j = A_2(\na n_0 + n_0\na V) + B_2\cdot(\na\vec n+\vec n\otimes\na V)
  + C_2(n_0,\vec n),
\end{align*}
where $A_j$, $B_j$, and $C_j$ are some functions depending also on the (given)
pseudo-spin polarization and the direction of the local pseudo-magnetization,
and $G_j$ depends on $\vec n$ and $J_j$. We refer to Section \ref{sec.model2}
for details. Because of the cross-diffusion structure, the analysis of this
model is not immediate, and we solve this model only numerically (see Section
\ref{sec.num}).

Thanks to the decoupled structure of the model QSDE1, 
we are able to perform an analytical study. More precisely, 
we show in Section \ref{sec.anal} the global existence and uniqueness
of weak solutions, some entropy-dissipation properties, 
and the exponential long-time decay
of the spin vector. As mentioned above, the main challenge is the proof 
of $|\vec n|/n_0<1$. By the maximum principle, it is not difficult
to prove that $n_0$ is strictly positive. However, an application of the
maximum principle to the equation for the spin vector is less obvious.
Our idea is to show that $u=1-|\vec n|^2/n_0^2$ satisfies the equation
$$
  \pa_t u - \Delta u - \na(\log n_0+V)\cdot\na u = 2G[\vec n/n_0],
$$
where $G[\vec n/n_0]$ is some nonnegative function.
This simple structure comes from the fact that
certain antisymmetric terms in $A_0$ and $B_0$ cancel in this situation.
By Stampacchia's truncation method, we conclude that there exists a positive lower 
bound for $u$ which proves that $|\vec n|/n_0<1$.

Finally, we present in Section \ref{sec.num} some numerical results for
the models QSDE1 and QSDE2, applied to a simple ballistic diode in one space
dimension. The equations are discretized by a Crank-Nicolson finite-difference
method. We illustrate the behavior of the particle density
$n_0$ and the spin components $n_j$ for various instants of time and
the exponential convergence of the particle density to the steady state. 

The paper is organized as follows. Section \ref{sec.model} is concerned with
the derivation of the models QSDE1 and QSDE2. The model QSDE1 is analyzed in
Section \ref{sec.anal}. Finally, numerical experiments are presented in 
Section \ref{sec.num}.

%%%%%%%%%%%%%%%%%%%%%%%%%%%%%%%%%%%%%%%%%%%%%%%%%%%%%%%%%%%%%%%%%%%%%%%%%%%%%%

\section{Modeling}\label{sec.model}

\subsection{A kinetic model for graphene}\label{sec.kinetic}

We describe the kinetic model for the quantum transport in graphene associated
to the Hamiltonian $H+V\sigma_0$, where $H$ is given by \eqref{H}, and $V$ is
the electric potential. Let $w(x,p,t)$
denote the system Wigner distribution, depending on the position $x\in\R^2$, 
momentum $p\in\R^2$, and time $t\ge 0$. The Wigner function takes values in the
space of complex Hermitian $2\times 2$ matrices, which is an Hilbert space
with respect to the scalar product $(A,B)=\frac12\mbox{tr}(AB)$, where 
$\mbox{tr}(A)$ denotes the trace of the matrix $A$. The set of Pauli matrices
\begin{equation}\label{pauli}
  \sigma_0 = \begin{pmatrix} 1 & 0 \\ 0 & 1 \end{pmatrix}, \quad
  \sigma_1 = \begin{pmatrix} 0 & 1 \\ 1 & 0 \end{pmatrix}, \quad
  \sigma_2 = \begin{pmatrix} 0 & -i \\ i & 0 \end{pmatrix}, \quad
  \sigma_3 = \begin{pmatrix} 1 & 0 \\ 0 & -1 \end{pmatrix}
\end{equation}
is a complete orthonormal system on that space. Therefore, we can develop
the Wigner function $w$ in terms of the Pauli matrices, $w=\sum_{j=0}^3 w_j\sigma_j$,
where $w_j(x,p,t)$ are real-valued scalar functions. We set
$\vec w=(w_1,w_2,w_3)$, $\vec p=(p_1,p_2,0)$, $p=(p_1,p_2)$, 
$\vec\sigma=(\sigma_1,\sigma_2,\sigma_3)$, and we abbreviate
$\pa_t=\pa/\pa t$ and $\vec\na=(\pa/\pa x_1,\pa/\pa x_2,0)$. 
With this notation, we can write $w=w_0\sigma_0+\vec w\cdot\vec\sigma$.
By applying the Wigner transform to the Von-Neumann equation, associated
to the Hamilonian $H+V$, the following Wigner equations for the quantum
transport in graphene have been derived in \cite{ZaBa11}:
\begin{equation}\label{wigner}
\begin{aligned}
  \pa_t w_0 + \left(\frac{\vec p}{m}\cdot\vec\na\right) w_0 + v_F\vec\na\cdot\vec w
  + \theta_\hbar[V]w_0 &= \frac{g_0-w_0}{\tau_c}, \\
  \pa_t\vec w + \left(\frac{\vec p}{m}\cdot\vec\na\right)\vec w 
  + v_F\left(\vec\na w_0 + \frac{2}{\hbar}\vec w\wedge\vec p\right)
  + \theta_\hbar[V]\vec w &= \frac{\vec g-\vec w}{\tau_c},
\end{aligned}
\end{equation}
where $\hbar$ is the reduced Planck constant. The parameter $m$, which has the
dimension of a mass, appears in the Hamiltonian $H$; see \eqref{H}.
The expressions $((\vec p/m)\cdot\vec\na)w_j$ are originating from the quadratic term
in the Hamiltonian $H$. Compared to Formula (12) in \cite{ZaBa11},
we have allowed for BGK-type collision operators on the right-hand sides 
of \eqref{wigner} with the relaxation time $\tau_c$ and
the thermal equilibrium distribution $g = g_0\sigma_0 + \vec g\cdot\vec\sigma$ 
which is defined in Section \ref{sec.g}.
The pseudo-differential operator $\theta_\hbar[V]w$ is given by
$$
  (\theta_\hbar[V]w)(x,p) = \frac{i}{\hbar}\,\frac{1}{(2\pi)^2}
  \int_{\R^2}\int_{\R^2}\delta V(x,\xi)w(x,p')e^{-i(p-p')\cdot\xi}d\xi dp',
$$
with its symbol
$$
  \delta V(x,\xi) = V\left(x+\frac{\hbar}{2}\xi\right)
  - V\left(x-\frac{\hbar}{2}\xi\right).
$$

In order to derive macroscopic diffusive models, we perform a diffusion scaling.
We introduce a typical spatial scale $\widehat x$, time scale $\widehat t$,
momentum scale $\widehat p$, and potential scale $\widehat V$:
$$
  x\to \widehat x, \quad t\to \widehat t t, \quad p\to \widehat p p, \quad
  V\to \widehat V V,
$$
where the scales are related to
$$
  \frac{2v_F\widehat p}{\hbar} = \frac{\widehat V}{\widehat x\widehat p}, \quad
  \frac{2\widehat p v_F \tau_c}{\hbar} = \frac{\hbar}{2\widehat p v_F\widehat t}, \quad
  \widehat p = \sqrt{mk_BT}.
$$
Here, $T$ is the (constant) system temperature and $k_B$ the Boltzmann constant.
The third relation means that the typical value of the momentum equals to
the thermal momentum. Let $L$ denote the average distance which a particle travels 
with the Fermi velocity $v_F$ between two consecutive
collisions, i.e.\ $L=\tau_c v_F$. Then the first relation can be written as
$$
  \frac{\widehat x}{L} = \frac12\,
  \frac{\widehat V(\hbar/\tau_c)}{(\widehat p^2/m)(mv_F^2)}.
$$
Thus, the ratio of the typical length scale and the ``Fermi mean free path''
is assumed to be of the same order as the quotient of the electric/wave energies and
the kinetic/Fermi energies. The second relation
$$
  \frac{\widehat t}{\tau_c} 
  = \frac14\,\frac{(\hbar/\tau_c)^2}{(\widehat p^2/m)(mv_F^2)}
$$
means that the ratio of the typical time scale and the relaxation time
is of the same order as the quotient of the square of the wave energy and
the kinetic/Fermi energies. 

We introduce the semiclassical parameter $\eps$, the diffusion parameter $\tau$,
and the scaled Fermi speed $c$, given by
$$
  \eps = \frac{\hbar}{\widehat x\widehat p}, \quad
  \tau = \frac{2\widehat p v_F\tau_c}{\hbar}, \quad
  c = \sqrt{\frac{mv_F^2}{k_BT}}.
$$
We suppose the semiclassical hypothesis $\eps\ll 1$ and the so-called
low scaled Fermi speed hypothesis,
\begin{equation}\label{lsfs}
  \gamma := \frac{c}{\eps} = O(1) \quad\mbox{as }\eps\to 0.
\end{equation}

With the above scaling, equations \eqref{wigner} become
\begin{equation}\label{wigners}
\begin{aligned}
  \tau\pa_t w_0 + \frac{1}{2\gamma}
  (\vec p\cdot\vec\na) w_0 + \frac{\eps}{2}\vec\na\cdot\vec w
  + \theta_\eps[V]w_0 &= \frac{g_0-w_0}{\tau}, \\
  \tau\pa_t\vec w + \frac{1}{2\gamma}(\vec p\cdot\vec\na)\vec w 
  + \frac{\eps}{2}\vec\na w_0 + \vec w\wedge\vec p
  + \theta_\eps[V]\vec w &= \frac{\vec g-\vec w}{\tau}.
\end{aligned}
\end{equation}
The ``drift'' terms $(\eps/2)\vec\na\cdot\vec w$ and 
$(\eps/2)\vec\na w_0$ are of order $O(\eps)$, whereas the ``precession'' term
$\vec w\wedge\vec p$ is of order one. This means that we have chosen a time scale
which is of the same order as the magnitude of the precession period of the spin
around the current, which is smaller than the typical time scale of the
drift process.

%%%%%%%%%%%%%%%%%%%%%%%%%%%%

\subsection{Thermal equilibrium distribution}\label{sec.g}

We define now the thermal equilibrium distribution
$g=g_0\sigma+\vec g\cdot\vec\sigma$ 
using the minimum entropy principle. We introduce the (unscaled) quantum entropy by
$$
  A[w] = \int_{\R^2}\int_{\R^2}\mbox{tr}\left(w\left(\mbox{Log}(w) - 1 
  + \frac{h(p)}{k_B T}\right)\right)dxdp,
$$
where $\mbox{Log}(w) = \mbox{Op}_\eps^{-1}\log \mbox{Op}_\eps(w)$ is the so-called
quantum logarithm introduced by Degond and Ringhofer \cite{DeRi03},
$\mbox{Op}_\eps$ is the Weyl quantization, 
defined for any symbol $\gamma(x,p)$ and any test function $\psi$
by \cite[Chapter 2]{Fol89}
$$ 
  (\mbox{Op}_\eps(\gamma)\psi)(x) = \frac{1}{(2\pi\hbar)^2}
  \int_{\R^2}\int_{\R^2}\gamma\left(\frac{x+y}{2},p\right)\psi(y) 
  e^{i(x-y)\cdot p/\hbar}dy dp,
$$
and
$$
  h(p) = \frac{|p|^2}{2m}\sigma_0 + v_F(p_1\sigma_1 + p_2\sigma_2)
$$
is the symbol of the Hamiltonian $H$, i.e.\ $H=\mbox{Op}_\hbar(h)$.

According to the theory of Degond and Ringhofer \cite{DeRi03}, we define the
Wigner distribution at local thermal equilibrium related to the given functions
$n_0$ and $\vec n$ as the formal solution $g=g[n_0,\vec n]$ (if it exists) to the
problem
$$
  A\big[g[n_0,\vec n]\big] = \min_w\left\{A[w]: \int_{\R^2}w_0 dx = n_0, \quad
  \int_{\R^2}\vec w dx = \vec n\right\},
$$
where the minimum is taken over all Wigner functions with complex Hermitian values
and $w$ is decomposed according to $w=w_0\sigma_0+\vec w\cdot\vec\sigma$.
This problem can be solved formally by means of Lagrange multipliers; see
\cite[Section 3.2]{Zam12}. For scalar-valued Wigner functions, such problems
are studied analytically in \cite{MePi11}. Formally, the (scaled) solution is
given by
\begin{equation}\label{gn}
  g[n_0,\vec n] = \mbox{Exp}(-h_{A,B}), \quad
  h_{A,B} = \left(\frac{|p|^2}{2}+A\right)\sigma_0 + (c\vec p+\vec B)\cdot\vec\sigma,
\end{equation}
where $A=A(x,t)$ and $\vec B=\vec B(x,t)=(B_1,B_2,B_3)(x,t)$ are the Lagrange
multipliers determined by 
\begin{equation}\label{constr}
  \int_{\R^2}g[n_0,\vec n](x,p,t)dp = n_0(x,t), \quad
  \int_{\R^2}\vec g[n_0,\vec n](x,p,t)dp = \vec n(x,t),
\end{equation}
and $\mbox{Exp}(w) = \mbox{Op}_\eps^{-1}\exp \mbox{Op}_\eps(w)$ is the
quantum exponential \cite{DeRi03}.

We wish to find an approximate but explicit expression for $g$. To this end,
we expand the quantum exponential in terms of powers of $\eps$, using the semiclassical
and the low scaled Fermi speed hypotheses. The expansion follows the lines
of Section 3.4 in \cite{Zam12}. We obtain from \eqref{gn}:
$$
  g[n_0,\vec n] = \mbox{Exp}(a+\eps b), \quad
  a = -\left(\frac{|p|^2}{2}+A\right)\sigma_0 - \vec B\cdot\vec\sigma, \quad
  b = -\gamma\vec p\cdot\vec\sigma.
$$
Note that $a$ and $b$ are of order one, in view of \eqref{lsfs}. Employing formulas
(29), (37), and (38) of \cite{Zam12}, we deduce that $g=g^{(0)} + \eps g^{(1)}
+ O(\eps^2)$, where
\begin{equation}\label{g.exp}
\begin{aligned}
  g^{(0)} &= e^{-(A+|p|^2/2)}\left(\cosh|\vec{B}|\sigma_0 
  - \frac{\sinh|\vec{B}|}{|\vec{B}|}\vec{B}\cdot\vec{\sigma}\right), \\
  g^{(1)} &= \gamma e^{-(A+|p|^2/2)}\Bigg\{
  \frac{\sinh |\vec{B}|}{|\vec{B}|}(\vec{B}\cdot\vec{p})\sigma_0  \\ 
  &\phantom{xx}{}
  - \Bigg[\left(\left(\cosh|\vec{B}| - \frac{\sinh |\vec{B}|}{|\vec{B}|}\right)
  \frac{\vec{B}\otimes\vec{B}}{|\vec{B}|^2}
  + \frac{\sinh |\vec{B}|}{|\vec{B}|}I\right)\vec{p} \\ 
  &\phantom{xx}{}+ \left(\cosh|\vec{B}| - \frac{\sinh |\vec{B}|}{|\vec{B}|}\right)
  \frac{((\vec{p}\cdot\vec{\nabla}_x)\vec{B})\wedge\vec{B}}{2\gamma|\vec{B}|^2}
  \Bigg]\cdot\vec{\sigma}\Bigg\},
\end{aligned}
\end{equation}
where $I$ denotes the unit matrix in $\R^{3\times 3}$.
Finally, it remains to express the Lagrange multipliers $A$ and $\vec B$ in
terms of $n_0$ and $\vec n$ by means of the contraints \eqref{constr}.
We find after tedious but straightforward computations that
\begin{equation}\label{g.ab}
  e^{-A} = \frac{1}{2\pi}\sqrt{n_0^2 - |\vec{n}|^2} + O(\eps^2), \quad
  \vec{B} = -\frac{\vec{n}}{|\vec{n}|}\log
  \sqrt{\frac{n_0 + |\vec{n}|}{n_0 - |\vec{n}|}} + O(\eps^2),
\end{equation}
where $I$ denotes the unit matrix.
Equations \eqref{g.exp} and \eqref{g.ab} provide an explicit approximation of the
thermal equilibrium distribution. The functions $g_0$ and $\vec g$ are defined
by $g^{(0)}+\eps g^{(1)}=g_0\sigma_0+\vec g\cdot\vec\sigma$.

%%%%%%%%%%%%%%%%%%%%%%%%%%%%

\subsection{Derivation of the first model}\label{sec.model1}

We derive our first spinorial drift-diffusion model. We assume that both
the semiclassical parameter $\eps$ and the diffusion parameter $\tau$ 
in \eqref{wigners} are small and of the same order. We will perform the limit
$\tau\to 0$ and $\eps\to 0$, setting 
$$
  \lambda := \frac{c}{\tau} = \frac{\eps\gamma}{\tau} = O(1) \quad\mbox{as }
  \tau\to 0.
$$
{}From \eqref{wigners} follows that the lowest-order approximations of $w_0$ and 
$\vec w$ are $g_0$ and $\vec g$, respectively. In order to compute the
first-order approximation, we employ a Chapman-Enskog expansion of the Wigner
function $w=w_0+\vec w\cdot\vec\sigma$ around the equilibrium distribution $g$.
Inserting the expansions 
$w_0=g_0\sigma_0+\tau f_0$, $\vec w=\vec g+\tau\vec f$ into
\eqref{wigners} and performing the formal limit $\tau\to 0$, we infer that
$$
  f_0 = -\frac{1}{2\gamma}(\vec p\cdot\vec\na)g_0 + \vec\na V\cdot\vec\na_p g_0, \quad
  \vec f = -\frac{1}{2\gamma}(\vec p\cdot\vec\na)\vec g 
  + \vec\na V\cdot\vec\na_p\vec g - \vec g\wedge\vec p.
$$
Here, we have used the expansion $\theta_\eps[V]=-\vec\na V\cdot\vec\na_p + O(\eps)$.

The moment equations of \eqref{wigners} read as
\begin{equation}\label{mom.eqs}
\begin{aligned}
  \tau\pa_t n_0 + \frac{1}{2\gamma}\vec\na\cdot\int_{\R^2}\vec pw_0 dp
  + \frac{\eps}{2}\vec\na\cdot\vec n &= 0, \\
  \tau\pa_t\vec n + \frac{1}{2\gamma}\vec\na\cdot\int_{\R^2}\vec w\otimes\vec p dp
  + \frac {\eps}{2}\vec\na n_0 + \int_{\R^2}\vec w\wedge\vec p dp &= 0,
\end{aligned}
\end{equation}
since $\int_{\R^2}\theta_\eps[V]w_j dp=0$ for $j=0,1,2,3$. 
We need to compute the first-order moments
$$
  \int_{\R^2}p_j w_k dp, \quad j=1,2, \ k=0,1,2,3,
$$
in order to close the moment equations \eqref{mom.eqs}.
For this, we insert in these integrals
the expansions for $w_0$ and $\vec w$ as well as the expansions
\eqref{g.exp}-\eqref{g.ab}. After long but straightforward computations
and rescaling $x\to x/(2\gamma)$ and $V\to V/(2\gamma)$ in order to get rid 
of the factor $1/(2\gamma)$, we arrive at the expressions, up to terms of order
$O(\tau^2)$,
\begin{align*}
  \int_{\R^2}p_k w_0 dp &= -\tau(\lambda n_k + \partial_k n_0 + n_0\partial_k V), \\
  \int_{\R^2}p_k w_s dp &= -\tau(\lambda Q_{ks} + \partial_k n_s 
  + n_s \partial_k V + \eta_{s\ell k}n_\ell),
\end{align*}
where we have set $\pa_1=\pa/\pa x_1$, $\pa_2=\pa/\pa x_2$, $\pa_3=0$,
\begin{align*}
  Q_{ks} &= n_0\delta_{ks} 
  - \frac{1}{n_0}\Phi\left(\frac{|\vec{n}|}{n_0}\right)
  \big(|\vec{n}|^2\delta_{ks} - n_k n_s + \eta_{sj\ell}n_j
  \partial_k n_\ell\big), \\
  \Phi(y) &= y^{-2}\left(1 - \frac{2y}{\log(1+y) - \log(1-y)}\right),\quad 0 < y < 1,
\end{align*}
using Einstein's summation convention,
and $(\eta_{jk\ell})$ is the only antisymmetric 3-tensor which is invariant
under cyclic index permutations such that $\eta_{123}=1$. In other words, 
$\eta_{jk\ell}a_j b_k = (\vec a\wedge\vec b)_{\ell}$ for $\vec a$, $\vec b\in\R^3$.
Inserting these expressions into \eqref{mom.eqs}, we find the model QSDE1:
\begin{align}
  & \pa_t n_0 - \diver\,J_0 = 0, \quad J_0 = \na n_0 + n_0\na V, \label{1.eq.n0} \\
  & \pa_t n_j - \diver\,J_j = F_j, \quad
  F_j = \eta_{jk\ell}n_k\pa_\ell V - 2n_j + b_k[\vec n/n_0]\pa_k n_j
  - b_j[\vec n/n_0]\vec\na\cdot\vec n, \label{1.eq.nj} \\
  & J_{js} = \big(\delta_{j\ell} + b_k[\vec n/n_0]\eta_{jk\ell}\big)\pa_s n_\ell
  + n_j\pa_s V - 2\eta_{js\ell}n_\ell + b_k[\vec n/n_0](\delta_{jk}n_s-\delta_{js}n_k),
  \label{1.eq.j}
\end{align}
where $j$, $s=1,2,3$. The functions
$$
  b_k[\vec v] = \lambda \frac{v_k}{|\vec v|^2}
  \left(1 - \frac{2|\vec v|}{\log(1+|\vec v|)-\log(1-|\vec v|)}\right), \quad
  k=1,2,3,\ \vec v\in\R^3,\ 0<|\vec v|<1,
$$
satisfy $0<b_k[\vec v]<\lambda$ for all $0<|\vec v|<1$, 
$\lim_{|\vec v|\to 0}b_k[\vec v]=0$,
and the function $|\vec v|\mapsto \Phi(|\vec v|)=b_k[\vec v]/(\lambda v_k)$ 
is increasing. 
This allows us to set $b_k[0]=0$ such that $b_k$ is defined for all $0\le|\vec v|<1$. 

The above equations are complemented by the Poisson equation
\begin{equation}
  -\lambda_D^2\Delta V = n_0 - C(x) \label{1.eq.V}
\end{equation}
for the electric potential, where $\lambda_D>0$ is the scaled Debye length.

%%%%%%%%%%%%%%%%%%%%%%%%%

\subsection{Derivation of the second model}\label{sec.model2}

In the model \eqref{1.eq.n0}-\eqref{1.eq.j}, the particle density $n_0$
evolves independently from the spin vector $\vec n$. We will modify this model
in order to derive a fully coupled system by adding a ``pseudo-magnetic'' field
which is supposed to be able to interact with the charge carrier pseudo-spin.

Possanner and Negulescu \cite{PoNe11} consider a semiconductor subject to a 
magnetic field which interacts with the electron spin and they build a purely
semiclassical diffusive model for the particle density $n_0$ and the spin
vector $\vec n$ by a Chapman-Enskog expansion around the equilibrium
distribution. Instead of the relaxation-time model used in Section \ref{sec.model1}, 
we employ here the mass and spin conserving collision operator (49) of \cite{PoNe11},
$$
  Q(w) = P^{1/2}(g-w)P^{1/2}, 
$$
where $g$ is the equilibrium distribution, defined in Section \ref{sec.g}, and
$P=\sigma_0+\zeta\vec\omega\cdot\vec\sigma$ is the polarization matrix
with the pseudo-spin polarization $\zeta(x,t)$ of the scattering rate and
$\vec\omega(x,t)$ is the direction of the pseudo-magnetization (see
\cite[Section 4.1]{PoNe11}). The quantity $\zeta(x,t)\in(0,1)$ satisfies
$$
  s_\uparrow = \frac{1 + |\zeta(x,t)|}{1 - |\zeta(x,t)|}s_\downarrow,
$$
where $s_{\uparrow\downarrow}$ are the scattering rates of electrons in the
upper and lower band, respectively, and it holds $|\vec\omega(x,t)|=1$.
We introduce the collision operators $Q_0$ and $\vec Q$ by
$$
  Q(w) = Q_0(w)\sigma_0 + \vec Q(w)\cdot\vec\sigma.
$$

We start from the scaled Wigner equations
\begin{equation}\label{wigner2}
\begin{aligned}
  \tau\pa_t w_0 + \frac{1}{2\gamma}(\vec p\cdot\vec\na)w_0 
  + \frac{\eps}{2}\vec\na\cdot\vec w + \theta_\eps[V]w_0 
  &= \frac{1}{\tau}Q_0(w), \\
  \tau\pa_t \vec w + \frac{1}{2\gamma}(\vec p\cdot\vec\na)\vec w
  + \frac{\eps}{2}\vec\na w_0 + \vec w\wedge\vec p + \theta_\eps[V]\vec w
  + \tau\vec\omega\wedge\vec w &= \frac{1}{\tau}\vec Q(w).
\end{aligned}
\end{equation}
Compared with the Wigner system \eqref{wigners} in Section \ref{sec.model},
the second equation in \eqref{wigner2} contains the (heuristic)
term $\tau\vec\omega\wedge\vec w$,
which describes the ``precession'' of $\vec w$ around the local pseudo-magnetization.
We assume again that $\lambda:=\eps\gamma/\tau$ is of order one (as $\tau\to 0$)
and we perform a Chapman-Enskog expansion of the Wigner distribution
$w=w_0\sigma_0+\vec w\cdot\vec\sigma$. The result reads as follows:
\begin{equation}\label{CE2}
  w = g - \tau P^{-1/2}T[g]P^{-1/2},
\end{equation}
where
$$
  T[g] = \left(\frac{1}{2\gamma}\vec p\cdot\vec\na - \vec\na V\cdot
  \vec\na_p\right)g_0\sigma_0
  + \left(\left(\frac{1}{2\gamma}\vec p\cdot\vec\na - \vec\na V\cdot
  \vec\na_p\right)\vec g + \vec g\wedge\vec p\right)\cdot\vec\sigma.
$$
To compute $P^{-1/2}T[g]P^{-1/2}$, we employ the following lemma whose
proof is an elementary computation.

\begin{lemma}\label{lem.P}
For all Hermitian matrices $a=a_0\sigma_0 + \vec a\cdot\vec\sigma$, it holds
\begin{align*}
  P^{-1/2}aP^{-1/2} &= \frac{1}{1-\zeta^2}(a_0-\zeta\vec\omega\cdot\vec a)
  \sigma_0 \\
  &\phantom{xx}{}
  + \frac{1}{1-\zeta^2}\big[\zeta\vec\omega a_0 + (\vec\omega\otimes\vec\omega
  + \sqrt{1-\zeta^2}(I-\vec\omega\otimes\vec\omega))\vec a\big]\cdot\vec\sigma.
\end{align*}
\end{lemma}
Since the collision operator $Q(w)$ is conserving mass and spin, the moment equations
of \eqref{wigner2} become
\begin{equation*}
\begin{aligned}
  \tau\pa_t n_0 + \frac{1}{2\gamma}\vec\na\cdot\int_{\R^2}\vec pw_0 dp
  + \frac{\eps}{2}\vec\na\cdot\vec n &= 0, \\
  \tau\pa_t\vec n + \frac{1}{2\gamma}\vec\na\cdot\int_{\R^2}\vec w\otimes\vec p dp
  + \frac{\eps}{2}\vec\na n_0 + \int_{\R^2}\vec w\wedge\vec p dp
  + \tau\vec\omega\wedge\vec n &= 0.
\end{aligned}
\end{equation*}
The first-order moments $\int_{\R^2}p_k w_0 dp$ and $\int_{\R^2}p_k w_s dp$
can be calculated up to order $O(\eps^2)=O(\tau^2)$
by using \eqref{CE2} and Lemma \ref{lem.P}, leading to the 
spin diffusion model QSDE2:
\begin{equation}\label{2.eq.n}
  \pa_t n_0 = \diver\,J_0, \quad
  \pa_t n_j = \diver\,J_j + G_j, \quad j=1,2,3,
\end{equation}
where 
\begin{equation}\label{2.eq.j}
\begin{aligned}
  J_{0s} &= (1-\zeta^2)^{-1}\big((\pa_s n_0+n_0\pa_s V) - \zeta\omega_k
  (\pa_s n_k + n_k\pa_s V + \eta_{k\ell s}n_\ell)\big), \\
  J_{js} &= (1-\zeta^2)^{-1}\big[-\zeta\omega_j (\pa_s n_0 + n_0\partial_s V) \\
  &\phantom{xx}{}+ (\omega_j\omega_k + \sqrt{1-\zeta^2}
  (\delta_{jk} - \omega_j\omega_k))
  (\pa_s n_k + n_k\pa_s V + \eta_{k\ell s}n_\ell)\big], \\
  G_j &= \eta_{jks}(J_{ks} + n_k\omega_s)
  + \pa_s\big(b_k[\vec n/n_0](\eta_{jk\ell}\pa_s n_\ell + \delta_{jk}n_s
  - \delta_{js}n_k)\big) \\
  &\phantom{xx}{}+ b_s[\vec n/n_0]\pa_s n_j - b_j[\vec n/n_0]\vec\na\cdot\vec n.
\end{aligned}
\end{equation}
In contrast to the model QSDE1 \eqref{1.eq.n0}-\eqref{1.eq.j}, this system
is fully coupled. It is possible to show that the system is uniformly parabolic
if $\|\zeta\|_{L^\infty(0,T;L^\infty(\Omega))}<1$ but the presence of the 
cross-diffusion terms makes it hard to prove any $L^\infty$ bounds, in 
particular the bound $|\vec n|/n_0<1$. 

%%%%%%%%%%%%%%%%%%%%%%%%%%%%%%%%%%%%%%%%%%%%%%%%%%%%%%%%%%%%%%%%%%%%%%%%%%%%%%

\section{Analysis for the first model}\label{sec.anal}

In this section, we consider the model QSDE1 \eqref{1.eq.n0}-\eqref{1.eq.V}
in a bounded domain $\Omega\subset\R^2$ with $\pa\Omega\in C^{1,1}$.
For convenience, we recall the equations:
\begin{align}
  & \pa_t n_0 - \diver\,J_0 = 0, \quad J_0 = \na n_0 + n_0\na V, 
  \label{3.eq.n0} \\
  & \pa_t n_j - \diver\,J_j = F_j, \quad
  F_j = \eta_{jk\ell}n_k\pa_\ell V - 2n_j + b_k[\vec n/n_0]\pa_k n_j
  - b_j[\vec n/n_0]\vec\na\cdot\vec n, \label{3.eq.nj} \\
  & J_{js} = \big(\delta_{j\ell} + b_k[\vec n/n_0]\eta_{jk\ell}\big)\pa_s n_\ell
  + n_j\pa_s V - 2\eta_{js\ell}n_\ell + b_k[\vec n/n_0](\delta_{jk}n_s-\delta_{js}n_k),
  \label{3.eq.j} \\
  & -\lambda_D^2\Delta V = n_0-C(x)\quad\mbox{in }\Omega,\ t>0, \ j=1,2,3, \ s=1,2,
  \label{3.eq.V}
\end{align}
and the functions
$$
  b_k[\vec v] = \lambda \frac{v_k}{|\vec v|^2}
  \left(1 - \frac{2|\vec v|}{\log(1+|\vec v|)-\log(1-|\vec v|)}\right), \quad
  k=1,2,3,\ \vec v\in\R^3,\ 0<|\vec v|<1.
$$
We impose the following boundary and initial conditions:
\begin{align}
  n_0 = n_D, \quad \vec n = 0, \quad V = V_D &\quad\mbox{on }\pa\Omega,\ t>0, 
  \label{3.bc} \\
  n_0(0) = n^0_{I}, \quad \vec n(0) = \vec n_I  &\quad\mbox{in }\Omega. \label{3.ic}
\end{align}
Finally, we abbreviate $\Omega_T=\Omega\times(0,T)$.

\subsection{Existence of solutions}\label{sec.ex}

We impose the following conditions on the data:
\begin{align}
  & n_D\in H^1(0,T;H^2(\Omega))\cap H^2(0,T;L^2(\Omega))\cap
  L^\infty(0,T;L^\infty(\Omega)), \label{1.nD} \\
  & n^0_{I}\in H^1(\Omega), \quad \inf_{\Omega}n_I^0 > 0, \quad
  n_I^0 = n_D(0) \quad\mbox{on }\pa\Omega, \quad
  \inf_{\pa\Omega\times(0,T)}n_D > 0, \label{1.n0} \\
  & V_D\in L^\infty(0,T;W^{2,p}(\Omega))\cap H^1(0,T;H^1(\Omega)), 
  \quad C\in L^\infty(\Omega), \ C\ge 0\mbox{ in }\Omega, \label{1.VD}
\end{align}
for some $p>2$.
Under these assumptions, we are able to prove the existence of strong 
solutions $(n_0,V)$ to the drift-diffusion model \eqref{3.eq.n0} and \eqref{3.eq.V}.

\begin{theorem}\label{thm.n0}
Let $T>0$ and assume \eqref{1.nD}-\eqref{1.VD}. Then there exists a unique
solution $(n_0,V)$ to \eqref{3.eq.n0} and \eqref{3.eq.V} subject to the
initial and boundary conditions in \eqref{3.bc}-\eqref{3.ic} satisfying
\begin{align*}
  & n_0\in L^\infty(0,T;H^2(\Omega))\cap H^1(0,T;H^1(\Omega))\cap
  H^2(0,T;(H^1(\Omega))'), \\
  & 0<me^{-\mu t}\le n_0\le M\quad\mbox{in }\Omega,\ t>0, \quad
  V\in L^\infty(0,T;W^{1,\infty}(\Omega)),
\end{align*}
where $\mu=\lambda_D^{-2}$ and
$$
  M = \max\left\{\sup_{\pa\Omega\times(0,T)}n_D,\ \sup_\Omega n_I^0,\ \sup_\Omega C(x)
  \right\},
  \quad m = \min\left\{\inf_{\pa\Omega\times(0,T)}n_D,\ \inf_\Omega n_I^0\right\} > 0.
$$
\end{theorem}

\begin{proof}
The existence and uniqueness of a weak solution $(n_0,V)$ to \eqref{3.eq.n0}, \eqref{3.eq.V}, and \eqref{3.bc}-\eqref{3.ic} satisfying
$$
  n_0\ge 0, \quad n_0\in L^2(0,T;H^1(\Omega))\cap H^1(0,T;(H^1(\Omega))'), \quad
  V\in L^2(0,T;H^1(\Omega))
$$
is shown in \cite{Gaj85}, also see Section 3.9 in \cite{MRS90}. 
It remains to prove the regularity assertions.

First, we show that $n_0$ is bounded from above and below.
Employing $(n_0-M)^+=\max\{0,n_0-M\}\in L^2(0,T;H^1_0(\Omega))$ 
as a test function in the weak formulation
of \eqref{3.eq.n0} and using \eqref{3.eq.V}, we find that
\begin{align*}
  \frac12\,\frac{d}{dt}\int_\Omega & \big((n_0-M)^+\big)^2 dx 
  + \int_\Omega|\na(n_0-M)^+|^2 dx
  = -\int_\Omega n_0\na V\cdot\na(n_0-M)^+ dx \\
  &= -\frac12\int_\Omega\na V\cdot\na\big((n_0-M)^+\big)^2 dx
  - M\int_\Omega\na V\cdot\na(n_0-M)^+ dx \\
  &= -\frac{1}{2\lambda_D^2}\int_\Omega(n_0-C(x))\big((n_0-M)^+\big)^2 dx
  - \frac{M}{\lambda_D^2}\int_\Omega(n_0-C(x))(n_0-M)^+ dx \\
  &\le 0,
\end{align*}
since $n_0-C(x)\ge 0$ on $\{n_0>M\}$, by the definition of $M$. This implies that $(n_0-M)^+=0$ and hence, $n_0\le M$ in $\Omega$, $t>0$.

Next, we employ the test function $(n_0-me^{-\mu t})^-=\min\{0,n_0-me^{-\mu t}\}$ 
in the weak formulation of \eqref{3.eq.n0}:
\begin{align*}
  \frac12\,\frac{d}{dt}\int_\Omega & \big((n_0-me^{-\mu t})^-\big)^2 dx
  + \int_\Omega|\na(n_0-me^{-\mu t})^-|^2 dx \\
  &= -\int_\Omega (n_0-me^{-\mu t})\na V\cdot\na(n_0-me^{-\mu t})^- dx
  - me^{-\mu t}\int_\Omega\na V\cdot\na(n_0-me^{-\mu t})^- dx \\
  &\phantom{xx}{}+ \mu me^{-\mu t}\int_\Omega(n_0-me^{-\mu t})^- dx \\
  &= -\frac{1}{2\lambda_D^2}\int_\Omega(n_0-C(x))\big((n_0-me^{-\mu t})^-\big)^2 dx \\
  &\phantom{xx}{}- \frac{1}{\lambda_D^2}\int_\Omega(n_0-C(x))(n_0-me^{-\mu t})^- dx
  + \mu me^{-\mu t}\int_\Omega(n_0-me^{-\mu t})^- dx \\
  &\le \frac{1}{2\lambda_D^2}\|C\|_{L^\infty(\Omega)}
  \int_\Omega\big((n_0-me^{-\mu t})^-\big)^2 dx
  - \int_\Omega\left(\frac{1}{\lambda_D^2}-\mu\right)
  me^{-\mu t}(n_0-me^{-\mu t})^- dx,
\end{align*}
since we integrate over $\{n_0<me^{-\mu t}\}$.
By the definition of $\mu$, the last integral vanishes. Then, the Gronwall lemma
implies that $(n_0-me^{-\mu t})^-=0$ and hence, $n_0\ge me^{-\mu t}$.

The above bounds show that the right-hand side of the Poisson equation
is an element of $L^\infty(0,T;L^\infty(\Omega))$. 
Then, by elliptic regularity, $V\in L^\infty(0,T;W^{2,p}(\Omega))$, where
$p>2$ is given in \eqref{1.VD}. Since
$W^{2,p}(\Omega)\hookrightarrow W^{1,\infty}(\Omega)$ (we recall
that $\Omega\subset\R^2$), it follows that
$\na V\in L^\infty(0,T;L^\infty(\Omega))$. Consider 
$$
  -\lambda_D^2\Delta\pa_t V = \pa_t n_0 \quad\mbox{in }\Omega, \quad
  \pa_t V = \pa_t V_D\quad\mbox{on }\pa\Omega.
$$
The right-hand side of this equation satisfies $\pa_t n_0
\in L^2(0,T;(H^1(\Omega))')$. Hence, $\pa_t V\in L^2(0,T;H^1(\Omega))$.

Finally, we prove the higher regularity for $n_0$. For this, we consider
the equation satisfied by $\rho=n_0-n_D$:
$$
  \pa_t\rho - \diver(\na\rho + \rho\na V) = f\mbox{ in }\Omega,\ t>0, \quad
  \rho=0\mbox{ on }\pa\Omega, \quad \rho(\cdot,0)=n_I^0-n_D(\cdot,0),
$$
where $f=-\pa_t n_D+\diver(\na n_D+n_D\na V)$. 
We use the following result: If $f$, $\pa_t f\in L^2(\Omega_T)$ and
$\rho(\cdot,0)\in H^2(\Omega)\cap H_0^1(\Omega)$ then $n_0\in C^0([0,T];H^2(\Omega))$,
$\pa_t n_0\in L^2(0,T;H^1(\Omega))$, $\pa_t^2 n_0\in L^2(0,T;(H^1(\Omega))')$
(see, e.g., \cite[Theorem 1.3.1]{Zhe95}).

The boundedness of $\na V$ implies that $f\in L^2(\Omega_T)$.
This gives the regularity $\pa_t n_0\in L^2(\Omega_T)$. As a consequence,
$V\in H^1(0,T;H^2(\Omega))$.
Our assumptions on the data show that $-\pa_t^2 n_D + \Delta\pa_t n_D\in
L^2(\Omega_T)$, and it remains to prove that $\pa_t\diver(n_D\na V)
\in L^2(\Omega_T)$. Now,
$$
  \pa_t\diver(n_D\na V) = \pa_t\na n_D\cdot\na V + \na n_D\cdot\na\pa_t V
  + \pa_t n_D\Delta V + n_D\Delta\pa_t V.
$$
The first term on the right-hand side lies in $L^2(\Omega_T)$
since $\na V\in L^\infty(0,T;L^\infty(\Omega))$. Furthermore, 
$\na n_D\in L^\infty(0,T;H^1(\Omega))$ and $\na\pa_t V\in L^2(0,T;H^1(\Omega))$
from we conclude that the second term is in $L^2(\Omega_T)$. In a similar
way, this property can be verified for the third and fourth terms.
This shows the claim and the regularity statements for $n_0$.
\end{proof}

Next, given $(n_0,V)$ as the solution to \eqref{3.eq.n0} and \eqref{3.eq.V},
we prove the existence of a solution $\vec n$ to \eqref{3.eq.nj}
with the corresponding boundary and initial conditions in \eqref{3.bc}-\eqref{3.ic},
satisfying the bound $|\vec n|/n_0<1$.

\begin{theorem}\label{thm.vecn}
Let $(n_0,V)$ be the solution to \eqref{3.eq.n0}, \eqref{3.eq.V}, and
\eqref{3.bc}-\eqref{3.ic}, according to Theorem \ref{thm.n0}, 
and let $\vec n_I\in H_0^1(\Omega)^3$ satisfy
$$
  \sup_{x\in\Omega}\frac{|\vec n_I(x)|}{n_I^0(x)} < 1.
$$
Then there exists a weak solution $\vec n\in L^2(0,T;H^1(\Omega))^3
\cap H^1(0,T;H^{-1}(\Omega))^3$ to \eqref{3.eq.nj}-\eqref{3.eq.j} and
\eqref{3.bc}-\eqref{3.ic} satisfying
\begin{equation}\label{vecn1}
  \sup_{(x,t)\in\Omega_T}\frac{|\vec n(x,t)|}{n_0(x,t)} < 1.
\end{equation}
Furthermore, there exists at most one weak solution satisfying \eqref{vecn1}
and $\vec n\in L^\infty(0,T;$ $W^{1,4}(\Omega))^3$.
\end{theorem}

\begin{proof}
We prove first the existence of solutions to a truncated problem by applying
the Leray-Schauder fixed-point theorem. Let $0<\chi<1$ be a fixed parameter
and let $\phi_\chi\in C^0(\R)$ be a nonincreasing function satisfying
$\phi_\chi(y)=1$ for $y\le 1-\chi$ and $\phi_\chi(y)=0$ for $y\ge 1$.
Then define
$$
  b_k^\chi[\vec v] = \phi_\chi(|\vec v|)b_k[\vec v] \quad\mbox{for all }\vec v\in\R^3.
$$

{\em Step 1: Application of the fixed-point theorem.}
In order to define the fixed-point operator, let $\vec\rho\in L^2(\Omega_T)^3$
and $\sigma\in[0,1]$. We wish to solve the linear problem
\begin{equation}\label{eq.lin}
  \frac{d}{dt}\int_\Omega \vec n\cdot\vec z dx + a(\vec n,\vec z;t) = 0
  \quad\mbox{for all }\vec z\in H_0^1(\Omega)^3, \quad
  \vec n(0) = \sigma\vec n_I,
\end{equation}
where 
\begin{align*}
  a(\vec n,\vec z;t)
  &= \int_\Omega\big((\delta_{j\ell}+\sigma b_k^\chi[\vec\rho/n_0]\eta_{jk\ell})
  \pa_s n_\ell + n_j\pa_s V - 2\eta_{js\ell} n_\ell\big)\pa_s z_j dx \\
  &\phantom{xx}{}
  + \sigma\int_\Omega b_k^\chi[\vec\rho/n_0](\delta_{jk}n_s-\delta_{js}n_k)\pa_s z_j dx
  \\
  &\phantom{xx}{}
  - \int_\Omega\big(\eta_{jk\ell}n_k\pa_\ell V - 2n_j + \sigma b_s^\chi[\vec\rho/n_0]
  \pa_s n_j - \sigma b_j^\chi[\vec\rho/n_0]\pa_s n_s\big)z_j dx,
\end{align*}
for $\vec z\in H_0^1(\Omega)^3$. The bilinear form $a:H_0^1(\Omega)^3\times
H_0^1(\Omega)^3\to\R$ is continuous since $|b_k^\chi[\vec\rho/n_0]|\le \lambda$ and
$|\na V|\in L^\infty(0,T;L^\infty(\Omega))$. Furthermore, using the antisymmetry
of $\eta_{jk\ell}$,
\begin{align*}
  a(\vec n,\vec n;t)
  &= \int_\Omega\big(\|\na\vec n\|^2 + n_j\pa_s V\pa_s n_j 
  - 2\eta_{js\ell} n_\ell\pa_s n_j\big) \\
  &\phantom{xx}{}
  + \sigma\int_\Omega b_k^\chi[\vec\rho/n_0](\delta_{jk}n_s-\delta_{js}n_k)\pa_s n_j dx
  \\
  &\phantom{xx}{}
  + \int_\Omega\big(2n_j - \sigma b_s^\chi[\vec\rho/n_0]
  \pa_s n_j + \sigma b_j^\chi[\vec\rho/n_0]\pa_s n_s\big)n_j dx,
\end{align*}
where $\|\na\vec n\|^2=\sum_{j,k=1}^3(\pa_j n_k)^2$. 
All the terms on the right-hand side
can be written as a product of $n_j$, $\pa_k n_\ell$, and possibly an $L^\infty$
function. Note that the only term, which does not have this structure,
$b_k^\chi\eta_{jk\ell}\pa_s n_\ell\pa_s n_j$, vanishes because of the
antisymmetry of $\eta_{jk\ell}$.
Therefore, the H\"older and Cauchy-Schwarz inequalities yield
$$
  a(\vec n,\vec n;t) \ge \frac12\|\vec n\|_{H^1(\Omega)}^2 
  - c\|\vec n\|_{L^2(\Omega)}^2
$$
for some constant $c>0$ which depends on the $L^\infty$ norm of $\na V$.
Hence, there exists a unique weak solution $\vec n\in L^2(0,T;H^1_0(\Omega))^3\cap
H^1(0,T;H^{-1}(\Omega))^3$ to the linear problem \eqref{eq.lin}
\cite[Corollary 23.26]{Zei90}. Moreover, there exists a constant $c>0$ independent
of $\rho$ and $\sigma$ such that
\begin{equation}\label{est.vecn}
  \|\vec n\|_{L^2(0,T;H^1(\Omega))^3} + \|\pa_t\vec n\|_{L^2(0,T;H^{-1}(\Omega))^3}
  \le c.
\end{equation}

This defines the fixed-point operator $F:L^2(\Omega_T)^3\times[0,1]\to L^2(\Omega_T)^3$,
$F(\vec\rho,\sigma)=\vec n$. We note that $F(\vec\rho,0)=0$.

Next, we show that $F$ is continuous. Let $(\vec\rho^{(k)})\subset L^2(\Omega_T)^3$
and $(\sigma^{(k)})\subset\R$ such that $\vec\rho^{(k)}\to\vec\rho$ in
$L^2(\Omega_T)^3$ and $\sigma^{(k)}\to\sigma$ as $k\to\infty$. Since $b_k^\chi$ is
bounded, it follows that $b_j^\chi[\rho^{(k)}/n_0]\to b_j^\chi[\rho/n_0]$
in $L^r(\Omega_T)$ for all $r<\infty$. Let $n^{(k)}=F(\rho^{(k)},\sigma^{(k)})$.
The uniform estimate \eqref{est.vecn} shows that, up to a subsequence,
\begin{align*}
  \vec n^{(k)}\rightharpoonup \vec n &\quad\mbox{weakly in }L^2(0,T;H^1(\Omega))^3
  \mbox{ and in }H^1(0,T;H^{-1}(\Omega))^3, \\
  \vec n^{(k)}\to \vec n &\quad\mbox{strongly in }L^2(\Omega_T)^3,
\end{align*}
since the embedding 
$L^2(0,T;H^1(\Omega))\cap H^1(0,T;H^{-1}(\Omega))\hookrightarrow L^2(\Omega_T)$
is compact, by the Aubin lemma. These convergence results are sufficient to
perform the limit $k\to\infty$ in the weak formulation of \eqref{eq.lin}
with $\vec n^{(k)}$ instead of $\vec n$ and $\sigma^{(k)}$ instead of $\sigma$. 
The limit equation shows that
$\vec n=F(\vec\rho,\sigma)$. As the solution to the linear problem is unique, the
convergence $\vec n^{(k)}\to\vec n$ in $L^2(\Omega_T)^3$
holds for the whole sequence, and hence, $S$ is continuous. Furthermore, by the
Aubin lemma, $F$ is compact. Finally, let $\vec n$ be a fixed point of
$F(\cdot,\sigma)$. By the boundedness of $b_j^\chi$ and estimate \eqref{est.vecn},
we find uniform estimates for $\vec n$ in $L^2(\Omega_T)^3$. Therefore, we can apply
the fixed-point theorem of Leray-Schauder which yields a weak solution to the 
truncated problem \eqref{eq.lin} with $\vec\rho$ replaced by $\vec n$.

{\em Step 2: $L^\infty$ bounds for $\vec n$.}
We prove that the solution to \eqref{eq.lin} is bounded in $\Omega_T$.
To this end, we define the function $\psi=\sqrt{1+|\vec n|^2}$.
Then, in the sense of distributions, 
$$
  \pa_t\psi = \frac{\pa_t\vec n\cdot\vec n}{\psi}
  = \frac{1}{2\psi}\big(\Delta(|\vec n|^2) + 2\diver(|\vec n|^2\na V)
  - \na V\cdot\na|\vec n|^2 - 2G[\vec n]\big)
$$
where 
\begin{equation}\label{G}
  G[\vec n]= \|\na\vec n\|^2 + 2\vec n\cdot\textrm{curl}\,\vec n + 2|\vec n|^2.
\end{equation} 
Inserting the identities
\begin{align*}
  \frac{1}{2\psi}\Delta(|\vec n|^2) &= \Delta\psi + \frac{1}{\psi}|\na\psi|^2, \\
  \frac{1}{\psi}\diver(|\vec n|^2\na V) &= \diver(\psi\na V) + \na V\cdot\na\psi
  - \frac{\Delta V}{\psi}, \\
  \frac{1}{2\psi}\na V\cdot\na|\vec n|^2 &= \na V\cdot\na\psi
\end{align*}
in the above equation for $\psi$, we deduce that
\begin{equation}\label{eq.psi}
  \pa_t\psi - \diver(\na\psi+\psi\na V) = - \psi^{-1}\Delta V 
  - \psi^{-1}\big(G[\vec n]-|\na\psi|^2\big).
\end{equation}

Since $|\textrm{curl}\,\vec v|^2 \le 2\|\na\vec v\|^2$ for all 
$\vec v\in H_0^1(\Omega)^3$, Young's inequality gives
\begin{equation}\label{G.pos}
  G[\vec v] \ge \|\na\vec v\|^2 - \left(2|\vec v|^2 + \frac12|\textrm{curl}\,\vec v|^2
  \right) + 2|\vec v|^2 \ge 0.
\end{equation}
Elementary computations show that 
$$
  G[\vec n]-|\na\psi|^2 = \psi^2 G\left[\frac{\vec n}{\psi}\right] 
  + \frac{|\na(\psi^2)|^2}{2\psi^4} \ge 0.
$$
Hence, we can estimate \eqref{eq.psi} by
$$
  \pa_t\psi - \diver(\na\psi+\psi\na V) \le \frac{1}{\lambda_D^2\psi}(n_0-C(x))
  \le \frac{n_0}{\lambda_D^2\psi} \le \frac{M}{\lambda_D^2}.
$$
By the maximum principle, $\psi\le c(T)$ in $\Omega_T$, where $c(T)>0$ depends
on the end time $T>0$. Taking into account that $\psi\ge 1$ by definition,
we conclude that $\psi\in L^\infty(0,T;L^\infty(\Omega))$ and hence,
$|\vec n|\in L^\infty(0,T;L^\infty(\Omega))$.

{\em Step 3: Proof of $|\vec n|/n_0<1$.} We show that there exists $0<\kappa<1$
such that $|\vec n|/n_0\le\kappa<1$. 
Then, choosing $\chi>0$ sufficiently small, we can remove the
truncation obtaining a solution to the original problem.

Let $u=1-|\vec n|^2/n_0^2$. Note that this function is well defined since $n_0$
is strictly positive, see Theorem \ref{thm.n0}. A tedious computation shows that
$u$ solves 
$$
  \pa_t u - \Delta u = \na(\log n_0+V)\cdot\na u + 2G[\vec n/n_0].
$$
in the sense of distributions, where $G$ is defined in \eqref{G}. We prove a lower
bound for $u$ by testing the weak formulation of this equation by 
$U:=(u-k)^-=\min\{0,u-k\}\in L^2(0,T;H_0^1(\Omega))$, where
$k=\min\{\inf_{\pa\Omega\times(0,T)}u,\inf_{\Omega\times\{0\}}u\}>0$.
\begin{equation}\label{u2}
  \frac12\,\frac{d}{dt}\int_\Omega U^2 dx + \int_\Omega|\na U|^2 dx
  = \int_\Omega U\na(\log n_0+V)\cdot\na Udx + 2\int_\Omega UG[\vec n/n_0]dx.
\end{equation}
Since $G(\vec n/n_0]\ge 0$, the last integral is nonpositive. The first integral
is estimated by employing the lower bound for $n_0$, obtained in Theorem \ref{thm.n0},
and applying Young's and H\"older's inequalities:
\begin{align*}
  \int\Omega U\na(\log n_0+V)\cdot\na Udx
  &\le \left(\inf_{\Omega_T}n_0\right)^{-1}\int_\Omega|U||\na n_0||\na U|dx \\
  &\le \frac{\eps}{2}\int_\Omega|\na U|^2 dx 
  + c(\eps)\int_\Omega |U|^2|\na n_0|^2 dx \\
  &\le \frac{\eps}{2}\int_\Omega|\na U|^2 dx 
  + c(\eps)\|\na n_0\|_{L^4(\Omega)}^2\|U\|_{L^2(\Omega)}^2.
\end{align*}
The Gagliardo-Nirenberg, H\"older, and Poincar\'e inequalities imply that
$$
  \|U\|_{L^4(\Omega)}^2 \le c\|U\|_{L^2(\Omega)}\|U\|_{H^1(\Omega)}
  \le \frac{\eps}{2}\|\na U\|_{L^2(\Omega)}^2 + c(\eps)\|U\|_{L^2(\Omega)}^2,
$$
where $\eps>0$.
By Theorem \ref{thm.n0}, $\na n_0$ is an element of $L^\infty(0,T;W^{1,4}(\Omega))$,
in view of the embedding $H^2(\Omega)\hookrightarrow W^{1,4}(\Omega)$.
Collecting the above estimates, we infer that
$$
  \int_\Omega U\na(\log n_0+V)\cdot\na Udx
  \le \eps\int_\Omega|\na U|^2 dx + c(\eps)\int_\Omega U^2 dx.
$$
Inserting this inequality into \eqref{u2} and choosing $\eps<2$, it follows that
$$
  \frac12\,\frac{d}{dt}\int_\Omega U^2 dx \le c(\eps)\int_\Omega U^2 dx.
$$
Then Gronwall's lemma and $U(\cdot,0)=0$ gives $U=0$ in $\Omega_T$ and 
hence, $u\ge k>0$ in $\Omega_T$.

{\em Step 4: Uniqueness of solutions.}
Let $\vec u$ and $\vec v$ be two solutions to \eqref{3.eq.nj} and 
\eqref{3.bc}-\eqref{3.ic} satisfying \eqref{vecn1} and
$\vec u\in L^\infty(0,T;W^{1,4}(\Omega))^3$.
Set $\vec w=\vec u-\vec v$. Taking the
difference of the equations satisfied by $\vec u$ and $\vec v$, respectively,
and employing $\vec w$ as a test function, we find that
\begin{align*}
  \frac12\,\frac{d}{dt} & \int_\Omega |\vec w|^2 dx
  + \int_\Omega\|\na\vec w\|^2 dx 
  \le \int_\Omega\big\{-w_j\pa_k w_j\pa_k V + 2\eta_{jk\ell}w_\ell\pa_k w_j \\
  &\phantom{xx}{}
  - (b_k[\vec u]-b_k[\vec v])(\eta_{jk\ell}\pa_s u_\ell + \delta_{jk}u_s
  - \delta_{js}u_k)\pa_s w_j
  - b_k[\vec v](\delta_{jk}w_s-\delta_{js}w_k)\pa_s w_j\big\}dx, \\
  &\phantom{xx}{}
  + \int_\Omega\big\{(\eta_{jk\ell}w_k\pa_\ell V-2w_j)w_j
  + [(b_s[\vec u]-b_s[\vec v])\pa_s u_j + b_s[\vec v]\pa_s w_j]w_j \\
  &\phantom{xx}{}
  - [(b_j[\vec u]-b_j[\vec v])\pa_s u_s + b_j[\vec v]\pa_s w_s]w_j\big\}dx.
\end{align*}
Thanks to the $L^\infty$ bounds on $\na V$, $\vec u$, and $\vec v$, 
we can estimate as follows:
\begin{align*}
  \frac12\,\frac{d}{dt} \int_\Omega |\vec w|^2 dx
  + \int_\Omega\|\na\vec w\|^2 dx 
  &\le c\int_\Omega\big(|\vec w|\,\|\na\vec w\| + \|\na\vec u\|\,|\vec w|^2
  + \|\na\vec u\|\,|\vec w|\,|\na\vec w\|\big)dx \\
  &\le \frac12\int_\Omega\|\na\vec w\|^2 dx 
  + c\|\na\vec u\|_{L^4(\Omega_T)}(1+\|\na\vec u\|_{L^4(\Omega_T)})
  \int_\Omega|\vec w|^2 dx,
\end{align*}
where $c>0$ is some generic constant.
Since $\vec w(0)=0$, the $W^{1,4}$ regularity for $\vec u$ and
Gronwall's lemma imply the assertion.
This finishes the proof.
\end{proof}

%%%%%%%%%%%%%%%%%%%%%%%%%%%%%%

\subsection{Entropy dissipation}\label{sec.entropy}

Let $(n_0,\vec n,V)$ be a solution to \eqref{3.eq.n0}-\eqref{3.eq.V},
\eqref{3.bc}-\eqref{3.ic} according to Theorems \ref{thm.n0} and \ref{thm.vecn}.
We assume that the boundary data is in thermal equilibrium, i.e.
\begin{equation}\label{bc.te}
  n_D = e^{-V_D}, \quad V = V_D, \quad \vec n = 0\quad\mbox{on }\pa\Omega,
\end{equation}
where $V_D=V_D(x)$ is time-independent.
In this subsection, we will show that the macroscopic entropy
\begin{align*}
  S(t) &= \int_\Omega\left(\frac12(n_0+|\vec n|)\big(\log(n_0+|\vec n|)-1\big)
  + \frac12(n_0-|\vec n|)\big(\log(n_0-|\vec n|)-1\big)\right. \\
  &\phantom{xx}{}
  + (n_0-C(x))V - \left.\frac{\lambda_D^2}{2}|\na V|^2\right)dx
%  + \frac{\lambda_D^2}{2}|\na(V-U)|^2 - Un_0\right)dx
\end{align*}
is nonincreasing in time. Note that $n_0<|\vec n|$ by Theorem \ref{thm.vecn}
such that $\log(n_0-|\vec n|)$ is well defined. 

The functional $S(t)$ can be derived as follows. Inserting the thermal equilibrium
distribution $g[n_0,\vec n]$ in the quantum entropy $A[w]$, defined in
Section \ref{sec.g}, and taking into account the electric energy contribution,
it follows that the total macroscopic free energy reads as
$$
  \widetilde S(t) 
  = A[g[n_0,\vec n]] - \int_\Omega\left(C(x)V 
  + \frac{\lambda_D^2}{2} |\na V|^2\right)dx.
$$
Then the expansion of $g[n_0,\vec n]$ (see \eqref{gn}, \eqref{g.exp}, and \eqref{g.ab}) yields the above formula for $\widetilde S(t)= S(t)+O(\eps^2)$.

\begin{proposition}
The entropy dissipation $-dS/dt$ can be written as
\begin{align*}
  \frac{dS}{dt}
  &= -\frac12\int_\Omega(n_0+|\vec n|)|\na(\log(n_0+|\vec n|)+V)|^2 dx \\
  &\phantom{xx}{}-\frac12\int_\Omega(n_0-|\vec n|)|\na(\log(n_0-|\vec n|)+V)|^2 dx \\
  &\phantom{xx}{}
  - \frac12\int_\Omega |\vec n|\log\left(\frac{n_0+|\vec n|}{n_0-|\vec n|}\right)
  G[\vec n/|\vec n|]dx \le 0,
\end{align*}
where $G$ is defined in \eqref{G}.
\end{proposition}

Note that in the drift-diffusion model without spin contribution, i.e.\
$|\vec n|=0$, we recover the standard entropy dissipation term
$\int_\Omega n_0|\na(\log n_0+V)|^2 dx$.

\begin{proof}
Taking the time derivative of
$S$, we find after some computations that
\begin{align}\label{dSdt}
  \frac{dS}{dt} &= \int_\Omega\bigg(
  \left(\frac12\log(n_0^2-|\vec n|^2)+V\right)\pa_t n_0
  + \frac12\log\left(\frac{n_0+|\vec n|}{n_0-|\vec n|}\right)\frac{\vec n}{|\vec n|}
  \cdot\pa_t\vec n \\
%  + \lambda_D^2\na(V-U)\cdot\na V_t - U\pa_t n_0\right)dx.
  &\phantom{xx}{}+ (n_0-C(x))\pa_t V - \lambda_D^2\na V\cdot\na\pa_t V \bigg)dx.
  \nonumber
\end{align}
To be precise, the second term in the integral has to be understood in the
sense of $L^2(0,T;H^{-1}(\Omega))$. Since $V_D$ does not depend on time,
we have $\pa_t V=0$ on $\pa\Omega$, $t>0$. Hence,
$$
  \int_\Omega(n_0-C(x))\pa_t V dx = -\lambda_D^2\int_\Omega\Delta V\pa_t V dx
  = \lambda_D^2\int_\Omega\na V\cdot\na\pa_t V dx,
$$
and thus, the last two terms in $dS/dt$ cancel. 

In order to compute the second term in $dS/dt$, we observe that
$$
  \vec n\cdot\pa_t\vec n = \frac12\Delta(|\vec n|^2)
  + \diver(|\vec n|^2\na V) - \frac12\na(|\vec n|^2)\cdot\na V - G[\vec n],
$$
where $G$ is defined in \eqref{G}. Then, inserting this expression and 
\eqref{1.eq.n0} into \eqref{dSdt} and integrating by parts, we infer that
\begin{align*}
  \frac{dS}{dt} &= -\int_\Omega\bigg\{\na\left(\frac12\log(n_0^2-|\vec n|^2)+V\right)
  \cdot(\na n_0+n_0\na V) \\
  &\phantom{xx}{}
  + \frac12\na\left(\frac{1}{|\vec n|}\log\frac{n_0+|\vec n|}{n_0-|\vec n|}\right)
  \cdot\left(\frac12\na(|\vec n|^2) + |\vec n|^2\na V\right) \\
  &\phantom{xx}{}+ \frac{1}{2|\vec n|}\log\frac{n_0+|\vec n|}{n_0-|\vec n|}
  \left(\frac12\na(|\vec n|^2)\cdot\na V + G[\vec n]\right)\bigg\}dx.
\end{align*}
Since
$$
  \na\left(\frac{1}{|\vec n|}\log\frac{n_0+|\vec n|}{n_0-|\vec n|}\right)
  = \frac{1}{|\vec n|}\na\left(\log\frac{n_0+|\vec n|}{n_0-|\vec n|}\right)
  + \na\left(\frac{1}{|\vec n|}\right)\log\frac{n_0+|\vec n|}{n_0-|\vec n|},
$$
we can write
\begin{align*}
  \frac{dS}{dt} &= -\int_\Omega\bigg\{\na\left(\frac12\log(n_0^2-|\vec n|^2)+V\right)
  \cdot(\na n_0+n_0\na V) \\
  &\phantom{xx}{}
  + \frac{1}{2|\vec n|}\na\left(\log\frac{n_0+|\vec n|}{n_0-|\vec n|}\right)
  \cdot\left(\frac12\na(|\vec n|^2) + |\vec n|^2\na V\right)\bigg\}dx \\
  &\phantom{xx}{}
  -\frac12 \int_\Omega\log\frac{n_0+|\vec n|}{n_0-|\vec n|}
  \bigg\{\na\left(\frac{1}{|\vec n|}\right)\cdot
  \left(\frac12\na(|\vec n|^2) + |\vec n|^2\na V\right) \\
  &\phantom{xx}{}
  + \frac{1}{|\vec n|}\left(\frac12\na(|\vec n|^2)\cdot\na V 
  + G[\vec n]\right)\bigg\}dx.
\end{align*}
Straightforward computations show that
$$
  |\vec n|G\left[\frac{\vec n}{|\vec n|}\right]
  = \na\left(\frac{1}{|\vec n|}\right)\cdot
  \left(\frac12\na(|\vec n|^2) + |\vec n|^2\na V\right)
  + \frac{1}{|\vec n|}\left(\frac12\na(|\vec n|^2)\cdot\na V 
  + G[\vec n]\right),
$$
which allows us to reformulate the second integral. 
Together with some manipulations in the first integral, we obtain
\begin{align*}
  \frac{dS}{dt} &= -\frac12\int_\Omega\big\{\na(\log(n_0+|\vec n|)+V)\cdot
  \big(\na(n_0+|\vec n|)+(n_0+|\vec n|)\na V\big) \\
  &\phantom{xx}{}
  + \na(\log(n_0-|\vec n|)+V)\cdot
  \big(\na(n_0-|\vec n|)+(n_0-|\vec n|)\na V\big)\big\}dx \\
  &\phantom{xx}{}
  - \frac12\int_\Omega |\vec n|\log\left(\frac{n_0+|\vec n|}{n_0-|\vec n|}\right)
  G\left[\frac{\vec n}{|\vec n|}\right]dx.
\end{align*}
We observe that the expression 
$$
  |\vec n|\log\left(\frac{n_0+|\vec n|}{n_0-|\vec n|}\right)
  G\left[\frac{\vec n}{|\vec n|}\right] 
  = \frac{1}{n_0}\frac{1}{|\vec{n}|/n_0}\log\left(
  \frac{1 + |\vec{n}|/n_0}{1 - |\vec{n}|/n_0}\right)|\vec{n}|^2 
  G\left[\frac{\vec n}{|\vec n|}\right]
$$
is integrable because $\inf_{\Omega_T} n_0 > 0$, $\sup_{\Omega_T} |\vec{n}|/n_0 < 1$,
the map
$$
  (0,1-\eps)\to\R, \quad 
  x\mapsto \frac{1}{x}\log\left(\frac{1+ x}{1 - x}\right)
$$ 
is bounded for all $\eps > 0$ and $|\vec{n}|^2 G[\vec n/|\vec n|]\in L^1(\Omega_T)$.
Since
$$
  \na(n_0\pm |\vec n|)+(n_0\pm |\vec n|)\na V
  = (n_0\pm |\vec n|)\na\big(\log(n_0\pm |\vec n|) + V\big),
$$
this finishes the proof.
\end{proof}

%%%%%%%%%%%%%%%%%%%%%%%%%%%%%%%%%%%%%%%%%%%%%%%%%%%%%%%%%%%%%%%%%%%%%%%%%%%%%

\subsection{Long-time decay of the solutions}\label{sec.long}

Let $(n_0,\vec n,V)$ be a solution to \eqref{3.eq.n0}-\eqref{3.eq.V},
\eqref{3.bc}-\eqref{3.ic} according to Theorems \ref{thm.n0} and \ref{thm.vecn}.
We will show that, under suitable assumptions on the electric potential,
the spin vector converges to zero as $t\to\infty$.

\begin{theorem}\label{thm.long}
Let $2<p<\infty$. Then there exists a constant
$\eps_p>0$ such that if the condition 
$\|\na V\|_{L^\infty(0,T;L^\infty(\Omega))}\le \eps_p$ holds then
$$
  \|\vec n(\cdot,t)\|_{L^p(\Omega)} \le \|\vec n_I\|_{L^p(\Omega)} e^{-\kappa_p t},
  \quad t\in(0,T),
$$
for some constant $\kappa_p>0$ which depends on $p$, $\Omega$, and the $L^\infty$
norm of $\na V$. Furthermore, there exists $\eps_2>0$ such that
if $\|\Delta V\|_{L^\infty(0,T;L^\infty(\Omega))}\le\eps_2$ then
$$
  \|\vec n(\cdot,t)\|_{L^2(\Omega)} \le \|\vec n_I\|_{L^2(\Omega)} e^{-\kappa_2 t},
  \quad t\in(0,T),
$$
for some constant $\kappa_2>0$ which depends on $\Omega$ and the $L^\infty$ norm
of $\Delta V$.
\end{theorem}

Note that we may set $T=\infty$ yielding the desired convergence result.
For the proof of the second part of the theorem, we need the following lemma.

\begin{lemma}\label{lem.G}
There exists a constant $c_G>0$, depending only on $\Omega$, such that 
for all $\vec u\in H^1_0(\Omega)^3$,
$$
  \int_\Omega G[\vec u]dx \ge c_G\int_\Omega |\vec u|^2 dx,
$$
where $G$ is defined in \eqref{G}.
\end{lemma}

\begin{proof}
Let $\mu>0$ and consider the following bilinear form on $H_0^1(\Omega)^3$:
$$
  B_\mu(\vec u,\vec v) = \int_\Omega\big(\na \vec u:\na\vec v
  + \vec u\cdot\textrm{curl}\,\vec v + \vec v\cdot\textrm{curl}\,\vec u
  + (2+\mu)\vec u\cdot\vec v\big)dx,
$$
where $\na \vec u:\na\vec v=\sum_{j,k=1}^3\pa_j u_k\pa_j v_k$.
The bilinear form $B$ is symmetric, continuous, and coercive on $H_0^1(\Omega)^3$,
since, using \eqref{G.pos} and the Poincar\'e inequality,
$$
  B_\mu(\vec u,\vec u) = G[\vec u]+\mu\int_\Omega|\vec v|^2 dx
  \ge \mu\int_\Omega|\vec v|^2 dx \ge c\mu\|\vec v\|_{H_0^1(\Omega)^3}.
$$
Hence, by the Lax-Milgram lemma,
for all $\vec f\in L^2(\Omega)^3$, there exists a unique solution
$\vec u\in H_0^1(\Omega)^3$ to
$$
  B(\vec u,\vec v) = \int_\Omega \vec f\cdot\vec v dx, \quad \vec v\in H_0^1(\Omega)^3.
$$
This defines the linear operator $L:L^2(\Omega)^3\to L^2(\Omega)^3$, $L(\vec f)=\vec u$.
Since the range of $L$, $R(L)$, is a subset of $H_0^1(\Omega)^3$, which embeddes
compactly into $L^2(\Omega)^3$, $L$ is compact. Furthermore, $L$ is symmetric
(since $B$ is symmetric) and positive in the sense of $\int_\Omega L(\vec f)\cdot
\vec f dx > 0$ for all $\vec f\neq 0$ (since $B$ is coercive).
By the Hilbert-Schmidt theorem for symmetric compact operators
(see, e.g., \cite[Theorem VI.16]{ReSi80}), there exists a complete orthonormal
system $(\vec u^{(k)})$ of $L^2(\Omega)$ of eigenfunctions of $L$,
$$
  L(\vec u^{(k)}) = \lambda_k \vec u^{(k)}, \quad 0<\lambda_k\searrow 0
  \mbox{ as }k\to\infty.
$$
Note that $(\lambda_k)$ depends on $\mu$ since $L$ and $B$ do so.
In particular, $\vec u^{(k)}\in R(L)\subset H_0^1(\Omega)^3$ and
\begin{equation}\label{B.aux}
  B(\vec u^{(k)},\vec v) = \lambda_k^{-1}\int_\Omega \vec u^{(k)}\cdot\vec v dx
  \quad\mbox{for }\vec v\in H_0^1(\Omega)^3,\ k\in\N.
\end{equation}

We claim that $\lambda_1^{-1}>\mu$. Otherwise, if $\lambda_1^{-1}\le\mu$,
the definition of $B$ and \eqref{B.aux} yield
\begin{align*}
  \int_\Omega\big( & 
  \|\na\vec u^{(1)}\|^2 + 2\vec u^{(1)}\cdot\textrm{curl}\,\vec u^{(1)}
  + (2+\mu)|\vec u^{(1)}|^2\big)dx \\
  &= B(\vec u^{(1)},\vec u^{(1)})
  = \lambda_1^{-1}\int_\Omega|\vec u^{(1)}|^2 dx
  \le \mu\int_\Omega|\vec u^{(1)}|^2 dx.
\end{align*}
The terms containing $\mu$ cancel, which gives
$$
  \int_\Omega\big(\|\na\vec u^{(1)}\|^2 + 2\vec u^{(1)}\cdot\textrm{curl}\,\vec u^{(1)}
  + 2|\vec u^{(1)}|^2\big)dx \le 0.
$$
However, in view of \eqref{G.pos}, the integral is nonnegative and hence, 
it must vanish. Therefore,
\begin{equation}\label{aux1}
  \|\na\vec u^{(1)}\|^2 + 2\vec u^{(1)}\cdot\textrm{curl}\,\vec u^{(1)}
  + 2|\vec u^{(1)}|^2 = 0 \quad\mbox{a.e. in }\Omega.
\end{equation}
On the other hand, using $|\textrm{curl}\,\vec u^{(1)}|^2 \le 2\|\na\vec u^{(1)}\|^2$,
$$
  0 \ge \frac12|\textrm{curl}\,\vec u^{(1)}|^2 
  + 2\vec u^{(1)}\cdot\textrm{curl}\,\vec u^{(1)} + 2|\vec u^{(1)}|^2 
  = \left|\frac{1}{\sqrt{2}}\textrm{curl}\,\vec u^{(1)} + \sqrt{2}\vec u^{(1)}
  \right|^2,
$$
from which we infer that $\textrm{curl}\,\vec u^{(1)} = 2\vec u^{(1)}$
and, by \eqref{aux1}, $ \|\na\vec u^{(1)}\|^2 + 6|\vec u^{(1)}|^2 = 0$.
This implies that $\vec u^{(1)}=0$ which is absurd. Hence, $\lambda_1^{-1}>\mu$.

Now let $\vec u\in H_0^1(\Omega)^3\cap H^2(\Omega)^3$. 
We can decompose $\vec u$
in the orthonormal set $(\vec u^{(k)})$, $\vec u=\sum_{k\in\N}c_k \vec u^{(k)}$
for some $c_k\in\R$. It follows from \eqref{B.aux} and the orthogonality
of $(\vec u^{(k)})$ on $L^2(\Omega)^3$ that
$$
  B(\vec u,\vec u) \ge \sum_{k\in\N}c_k^2 B(\vec u^{(k)},\vec u^{(k)})
  = \sum_{k\in\N}c_k^2\lambda_k^{-1} \ge \lambda_1^{-1}\|\vec u\|_{L^2(\Omega)^3}^2.
$$
By a density argument, this inequality also holds for all 
$\vec u\in H_0^1(\Omega)^3$. Therefore,
$$
  \int_\Omega G[\vec u]dx = B_\mu(\vec u,\vec u) - \mu\int_\Omega|\vec u|^2 dx
  \ge (\lambda^{-1}-\mu)\|\vec u\|_{L^2\Omega)^3}^2,
$$
and the lemma follows with $c_G=\lambda_1^{-1}-\mu>0$.
\end{proof}

\begin{proof}[Proof of Theorem \ref{thm.long}.]
Let $p\ge 2$. Using the test function $p|\vec n|^{p-2}n_j\in L^2(0,T;H_0^1(\Omega))$ 
in \eqref{3.eq.nj}
and summing over $j=1,2,3$, we find after elementray computations that
\begin{align}\label{eq.np}
  \frac{d}{dt}\int_\Omega & |\vec n|^p dx 
  + \frac{4(p-2)}{p}\int_\Omega\big|\na|\vec n|^{p/2}\big|^2 dx \\
  &= -2(p-1)\int_\Omega|\vec n|^{p/2}\na V\cdot\na|\vec n|^{p/2} dx
  - p\int_\Omega |\vec n|^{p-2}G[\vec n]dx. \nonumber
\end{align}

Now, we distinguish two cases. First, let $p>2$. Employing 
Young's inequality with $\alpha>0$ and $G[\vec n]\ge 0$, we find that
\begin{align*}
   \frac{d}{dt}\int_\Omega & |\vec n|^p dx 
  + \frac{4(p-2)}{p}\int_\Omega\big|\na|\vec n|^{p/2}\big|^2 dx \\
  & \le \frac{(p-1)^2}{\alpha}\|\na V\|_{L^\infty(0,T;L^\infty(\Omega))}^2
  \int_\Omega |\vec n|^p dx 
  + \frac{\alpha}{2}\int_\Omega\big|\na|\vec n|^{p/2}\big|^2 dx.
\end{align*}
Then, choosing $\alpha=4(p-2)/p$ and employing the Poincar\'e inequality 
$$
  \int_\Omega\big|\na|\vec n|^{p/2}\big|^2 dx \ge C_P^2\int_\Omega|\vec n|^p dx,
$$
we infer that
$$
  \frac{d}{dt}\int_\Omega |\vec n|^p dx 
  \le \left(\frac{p(p-1)^2}{4(p-2)}\|\na V\|_{L^\infty(0,T;L^\infty(\Omega))}^2
  - \frac{2(p-2)}{p}C_P^2\right)\int_\Omega|\vec n|^{p} dx.
$$
This proves the first part after setting $\eps_p<\sqrt{8}(p-2)C_P/(p(p-1))$
and applying the Gronwall lemma.

For the second part, let $p=2$ in \eqref{eq.np}. By Lemma \ref{lem.G} and
integration by parts in the term containing the potential, we obtain
$$
  \frac{d}{dt}\int_\Omega |\vec n|^2 dx 
  = -\int_\Omega\Delta V|\vec n|^2 dx - 2\int_\Omega G[\vec n]dx
  \le \big(\|\Delta V\|_{L^\infty(0,T;L^\infty(\Omega))} - 2c_G\big)
  \int_\Omega |\vec n|^2 dx.
$$
With the choice $\eps_2<2c_G$ and the Gronwall lemma, the theorem follows.
\end{proof}

If the total space charge $n_0-C(x)=-\lambda_D^2\Delta V$ is positive,
we are able to prove that $\vec n(\cdot,t)$ converges to zero in the
$L^\infty$ norm.

\begin{proposition} 
Let $0<T\le\infty$. The following $L^\infty$ estimate holds:
$$
  \|\vec{n}(\cdot,t)\|_{L^\infty(\Omega)} 
  \le \|\vec{n}_I \|_{L^\infty(\Omega)} 
  \exp\big((\textstyle\sup_{\Omega\times (0,T)} \Delta V)t\big),\quad t\in(0,T).
$$
\end{proposition}

\begin{proof}
Let $2<p<\infty$ be arbitrary. From \eqref{eq.np} we deduce that, 
by integrating by parts,
$$
  \frac{d}{dt} \int_\Omega |\vec{n}|^p dx 
  \le (p-1)\sup_{\Omega_T} \Delta V\int_\Omega |\vec{n}|^p dx.
$$
Therefore 
$$
  \|\vec{n}(\cdot,t)\|_{L^p(\Omega)} 
  \le \exp\big((1-\tfrac{1}{p})
  (\textstyle\sup_{\Omega_T} \Delta V) t\big)\|\vec{n}_I\|_{L^p(\Omega)},
  \quad t\in(0,T).
$$
Passing to the limit $p\to\infty$ in this inequality yields the claim.
\end{proof}

%%%%%%%%%%%%%%%%%%%%%%%%%%%%%%%%%%%%%%%%%%%%%%%%%%%%%%%%%%%%%%%%%%%%%%%%%%%%%

\section{Numerical simulations}\label{sec.num}

In this section, we present some numerical results for the models
\eqref{1.eq.n0}-\eqref{1.eq.V} and \eqref{2.eq.n}-\eqref{2.eq.j}, \eqref{1.eq.V}
in one space dimension, $\Omega=(0,1)$. We choose the boundary conditions
$$
  n_0=C, \quad \vec n=0, \quad V=U \quad\mbox{on }\pa\Omega=\{0,1\},\ t>0,
$$
where $U(x)=V_Ax$ and $V_A=80$ is the scaled applied potential, 
and the initial conditions
$$
  n_0(x,0) = \exp(-V_{\rm eq}(x)), \quad \vec n(x,0)=0,
$$
where $V_{\rm eq}$ is the equilibrium potential, defined by
$$
  -\lambda_D^2\pa_{xx}^2 V_{\rm eq} = \exp(-V_{\rm eq})-C(x)\quad\mbox{in }\Omega,
  \quad V_{\rm eq}(0) = V_{\rm eq}(1)=0.
$$
We choose $\lambda_D^2=10^{-3}$.
The doping profile corresponds to that of a ballistic diode:
$$
  C(x)=C_{\rm min}\quad\mbox{for }\bar x<x<1-\bar x, \quad
  C(x)=1\quad\mbox{else},
$$
where $C_{\rm min}=0.025$ and $\bar x=0.2$. 
The pseudo-spin polarization and the direction of the local magnetization
are defined by
$$
  \zeta = 0.5, \quad \vec\omega = (0,0,1)^\top.
$$
Table \ref{table} shows the values of the units which allows for the
computation of the physical values from the scaled ones.

\begin{table}[ht]
\begin{tabular}{|l|l|}
\hline
space unit & $10^{-7}$\,m\\
\hline
time unit & $0.5\times 10^{-13}$\,s\\
\hline
voltage unit & $1.25\times 10^{-2}$\,V\\
\hline
particle density unit & $10^{17}$\,m$^{-2}$\\
\hline
current density unit & $2\times 10^{23}$\,m$^{-1}$ s$^{-1}$\\
\hline
\end{tabular}\vskip2mm
\caption{Units used for the numerical simulations.}\label{table}
\end{table}

Models QSDE1 \eqref{1.eq.n0}-\eqref{1.eq.V} and QSDE2 \eqref{2.eq.n}-\eqref{2.eq.j},
\eqref{1.eq.V} with the corresponding initial and boundary conditions
\eqref{3.bc}-\eqref{3.ic} are discretized with the Crank-Nicolson finite-difference
scheme and the space step $\triangle x = 10^{-2}$. 
The resulting nonlinear discrete ODE system is solved by using the 
Matlab routine ode23s. 

Since the initial spin vector is assumed to vanish, the particle density $n_0$,
computed from the model QSDE1, corresponds exactly to the particle density 
of the standard drift-diffusion model, and the spin vector vanishes for all time.
This situation is different in the model QSDE2 since the equations are fully
coupled. For the model QSDE2, Figure \ref{fig.moments} shows the particle density 
$n_0$ and the components $n_j$ of the spin vector versus position at various times.
The solution at $t=1$ corresponds to the steady state.
We observe a charge built-up of $n_0$ in the low-doped region of the diode.
The spin vector components vary only slightly in this region but their gradients
are significant in the high-doped regions close to the contacts.
Clearly, the components $n_j$ do not need to be positive and, in fact, they
even do not have a sign.

\begin{figure}[ht]
\begin{center}
\includegraphics[scale=0.77]{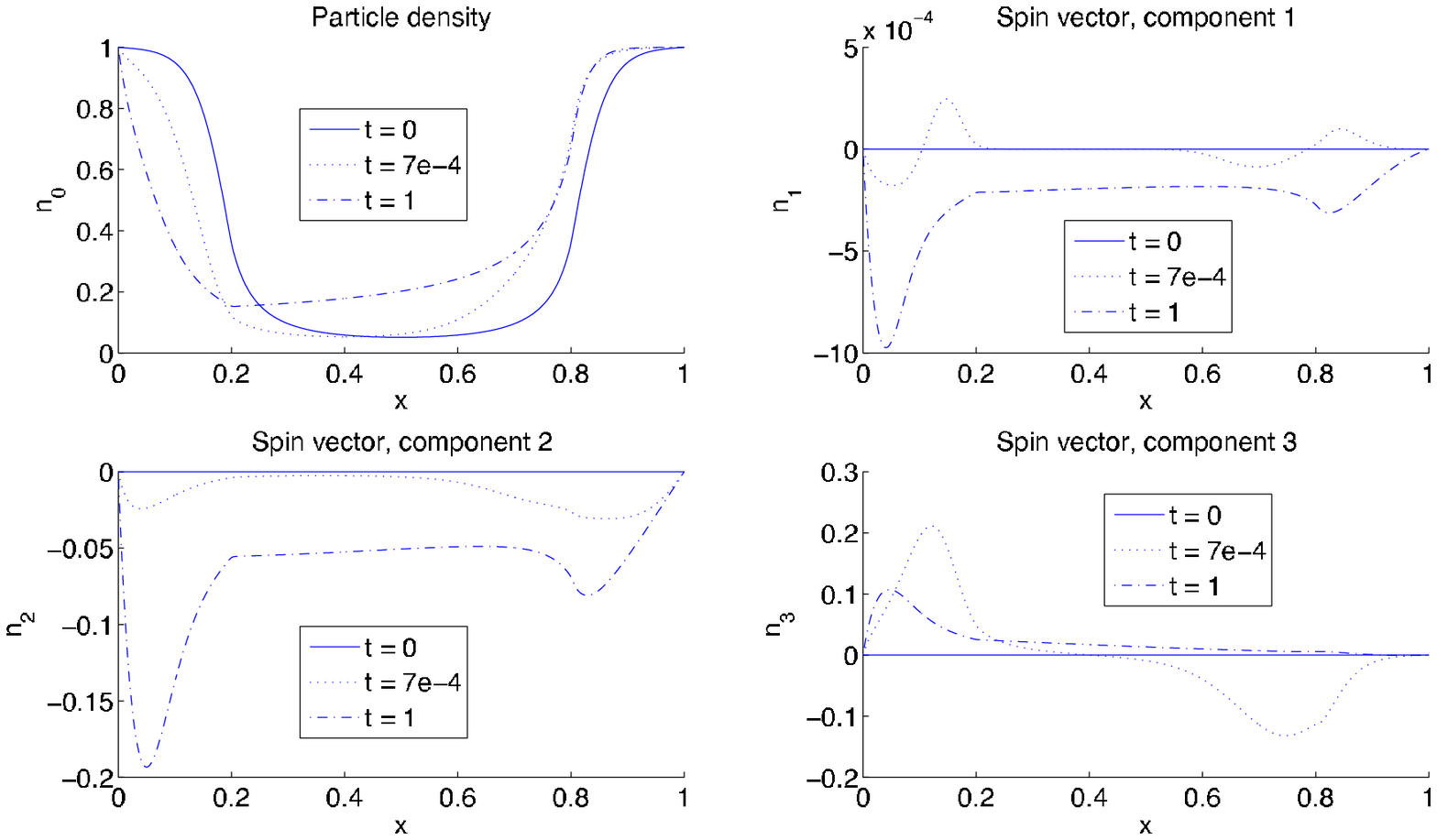}
\caption{Model QSDE2: Particle density and components of the spin vector versus position
at times $t=0$, $t=7\cdot 10^{-4}$, and $t=1$.}\label{fig.moments}
\end{center}
\end{figure}

The models QSDE1 and QSDE2 are well defined only if $|\vec n|/n_0<1$.
We plot this ratio in Figure \ref{fig.ratios} at various times for the model QSDE2.
In all the presented cases, the quotient stays below one. This indicates that
$b_k[\vec n/n_0]$ is well defined also in this model.

\begin{figure}[ht]
\centering
\includegraphics[scale=0.77]{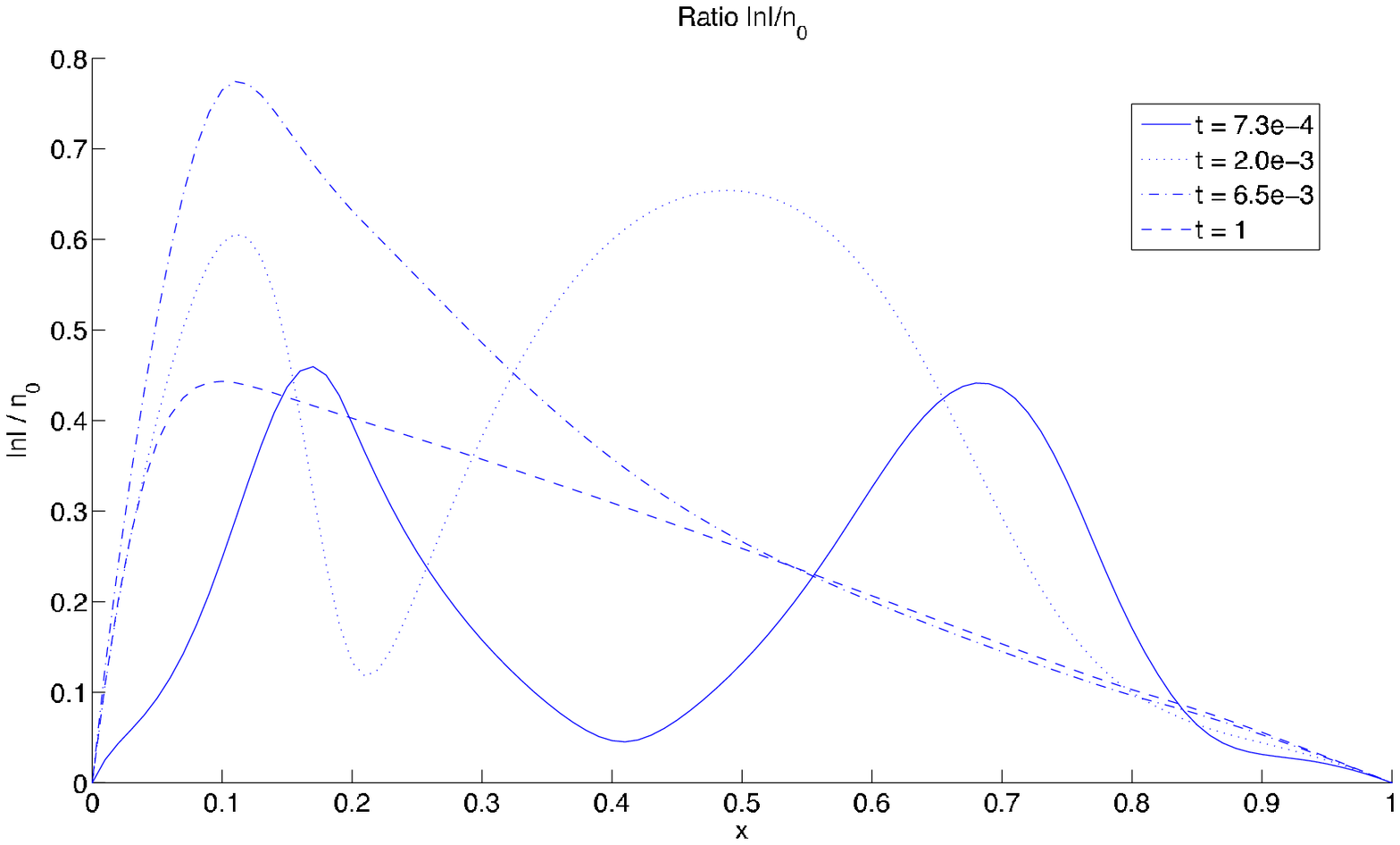}
\caption{Model QSDE2: Ratio $|\vec n|/n_0$ versus position at various times.}
\label{fig.ratios}
\end{figure}

We have shown in Theorem \ref{thm.long} that the spin vector of the model
QSDE1 converges to zero if the electric potential satisfies certain conditions.
In Figure \ref{fig.equi}, the relative difference 
$\|n_0(t)-n_0(\infty)\|_2/\|n_0(\infty)\|_2$
versus time is depicted (semilogarithmic plot), 
where $n_0(\infty)$ denotes the steady-state particle
density of model QSDE1 or QSDE2, respectively. The norm $\|\cdot\|_2$ is the
Euclidean norm. The stationary solution is approximated by $n_0(t^*)$ with
$t^*=1$. Whereas the decay of the solution to the model QSDE1 is numerically
of exponential type (in agreement with the theoretical results),
the decay for the model QSDE2 seems to be exponential only for small times.

\begin{figure}[ht]
\centering
\includegraphics[scale=0.77]{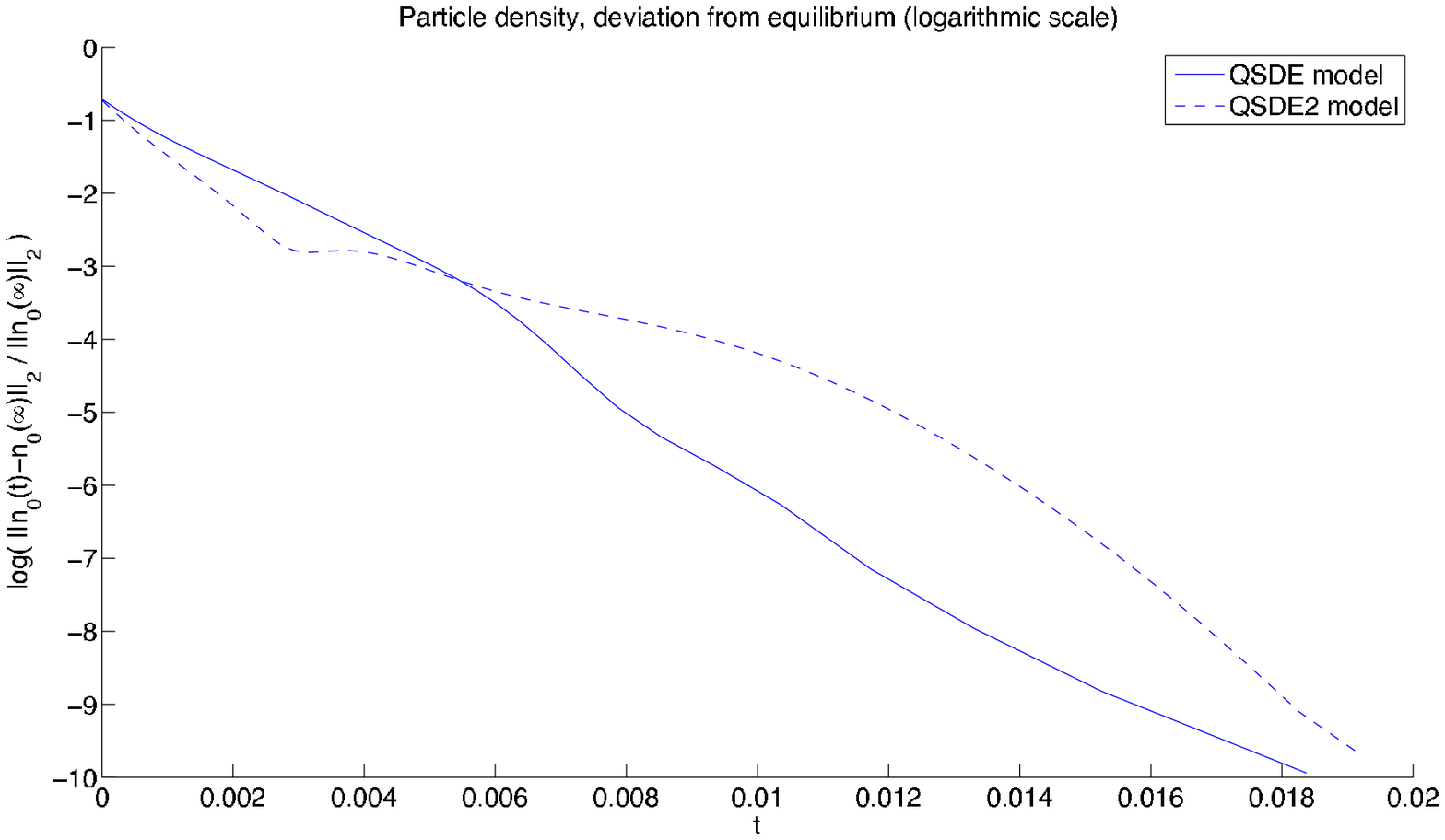}
\caption{Relative difference $\|n_0(t)-n_0(\infty)\|/\|n_0(t)\|$ versus time
(semilogarithmic plot) for the models QSDE1 (solid line) and QSDE2 (dashed line).}
\label{fig.equi}
\end{figure}

\begin{figure}[ht]
\centering
\includegraphics[scale=0.77]{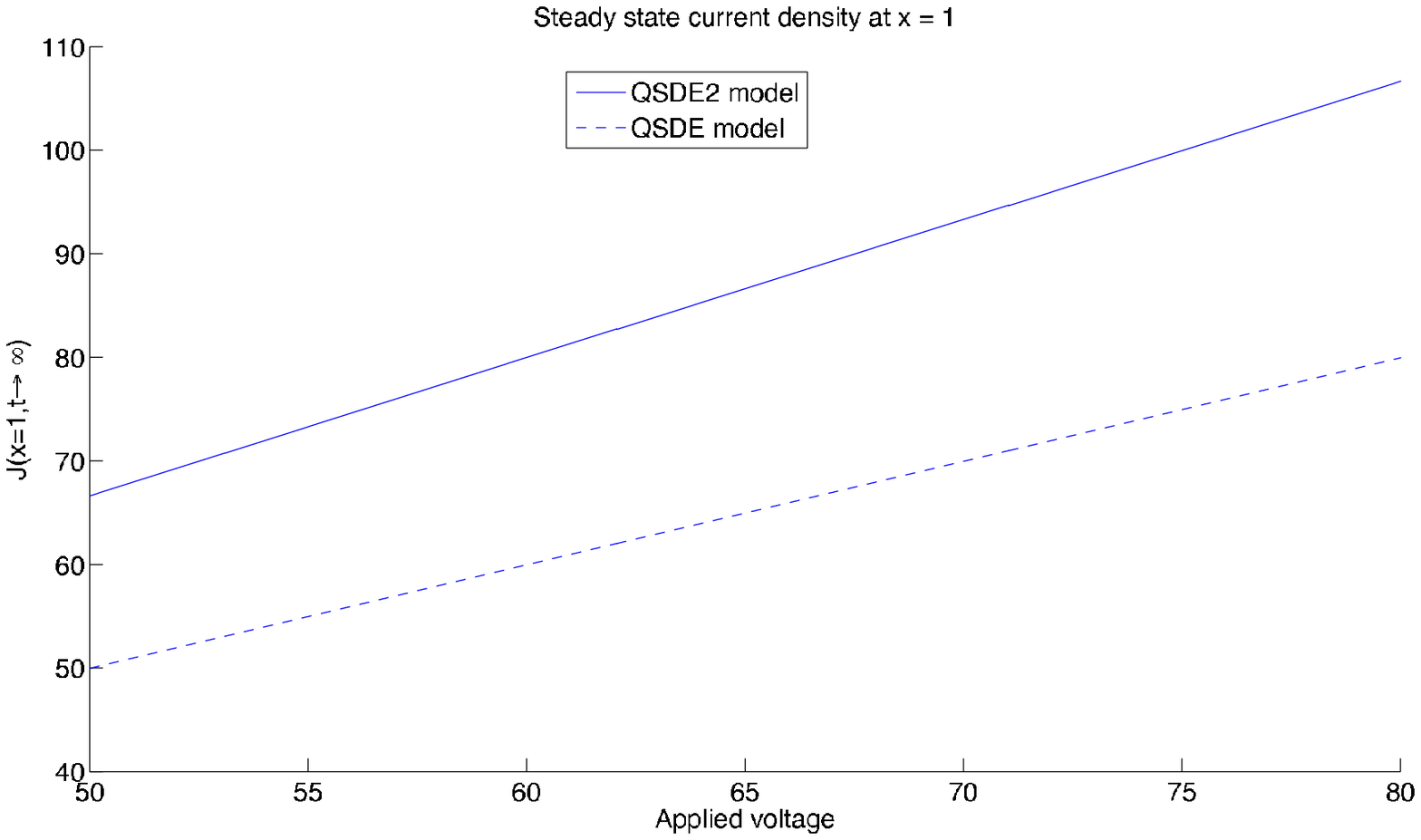}
\caption{Static current-voltage characteristics for the models QSDE1 and QSDE2.}
\label{fig.cv}
\end{figure}

In the final Figure \ref{fig.cv}, we present the current-voltage characteristics 
for the models QSDE1 and QSDE2, i.e.\ the relation between $J_0$ at $x=1$
and the applied bias $V_A$. The characteristics of model QSDE1 correspond
to the current-voltage curve of the standard drift-diffusion model. We observe
that the additional terms in the definition of $J_0$ lead to an
increase of the particle current density.

%%%%%%%%%%%%%%%%%%%%%%%%%%%%%%%%%%%%%%%%%%%%%%%%%%%%%%%%%%%%%%%%%%%%%%%%%%%%%

\end{document}